\documentclass[11pt,reqno]{amsart}
\usepackage{amscd,amssymb,amsmath}
\usepackage{bm}
\usepackage[mathscr]{euscript}
\usepackage{epsfig}
\usepackage{tikz}
\usetikzlibrary{calc}
\usepackage{tikz-cd}
\usepackage{citeref}
\usepackage{hyperref}

\title[]{On Frustration-Free Quantum Spin Models}

\author[D. P. Ojito]{Danilo Polo Ojito}
\address{Department of Mathematics\\ Pontificia Universidad Javeriana, Bogota, Colombia
	\href{danilopolo@javeriana.edu.co}{danilopolo@javeriana.edu.co}}

\author[E. Prodan]{Emil Prodan}

\address{Department of Physics and
 Department of Mathematical Sciences 
\\Yeshiva University, New York, NY 10016, USA \href{mailto:prodan@yu.edu}{prodan@yu.edu}
}

\author[T. Stoiber]{Tom Stoiber}

\address{Department of Mathematics, University of California, 
	\\Berkeley CA, 94703, USA \\
	\href{mailto:tstoiber@berkeley.edu}{tstoiber@berkeley.edu}}

\date{\today}

\newtheorem{theorem}{Theorem}[section]
\newtheorem{definition}[theorem]{Definition}
\newtheorem{proposition}[theorem]{Proposition}
\newtheorem{lemma}[theorem]{Lemma}
\newtheorem{corollary}[theorem]{Corollary}
\newtheorem{remark}[theorem]{Remark}
\newtheorem{example}[theorem]{Example}

\newcommand{\BM}{{\mathbb B}}
\newcommand{\CM}{{\mathbb C}}
\newcommand{\NM}{{\mathbb N}}
\newcommand{\RM}{{\mathbb R}}

\newcommand{\ZM}{{\mathbb Z}}

\newcommand{\Aa}{{\mathcal A}}

\newcommand{\Pp}{{\mathcal P}}

\newcommand{\Bb}{{\mathcal B}}
\newcommand{\Dd}{{\mathcal D}}

\newcommand{\Ww}{{\mathcal W}}

\newcommand{\Vv}{{\mathcal V}}
\newcommand{\Ss}{{\mathcal S}}

\newcommand{\Tt}{{\mathcal T}}

\newcommand{\Nn}{{\mathcal N}}

\newcommand{\Cc}{{\mathcal C}}

\newcommand{\Ll}{{\mathcal L}}

\newcommand{\Kk}{{\mathcal K}}
\newcommand{\Hh}{{\mathcal H}}

\begin{document}

\begin{abstract}
 The goal of our work is to characterize the landscape of the frustration-free quantum spin models over the Cayley graph of a finitely generated group $G$. This is achieved by establishing $G$-equivariant morphisms from the partially ordered space of frustration-free models to the partially ordered spaces 1) of hereditary $C^\ast$-algebras of the underlying UHF algebra of quasi-local observables, 2) of open projections in its double dual, and 3) of subsets of pure state space. Our main result consists of an intrinsic characterization of the images of these morphisms, which captures the essence of frustration-freeness and enables us to extend the concept to generic AF-$C^\ast$-algebras. Additionally, using well established facts about AF-$C^\ast$-algebras, we prove density theorems and provide intrinsic and optimal characterizations of frustration-free ground states.
\end{abstract}

\maketitle

%{\scriptsize \tableofcontents}

\section{Introduction} 

Consider quantum spin degrees of freedom distributed over the Cayley graph of a finitely generated group $G$. Then the quasi-local algebra of physical observables is the UHF algebra $\Ss^{\otimes G}$ with $\Ss$ a full matrix algebra. It accepts a natural $G$-action $\alpha$ and has $\Ss_\Lambda : = \Ss^{\otimes \Lambda}$ as a family of finite dimensional $C^\ast$-subalgebras, where $\Lambda$ samples ${\rm K}(G)$, the set of compact hence finite subsets of $G$. If $g\cdot g':=g' g^{-1}$ indicates the right action of $G$ on itself, and on ${\rm K}(G)$ as well, then $\alpha_g(\Ss_\Lambda)=\Ss_{g \cdot \Lambda}$. The left action of $G$ on ${\rm K}(G)$ will be indicated by $g\Lambda$ (hence no dot). The spaces of states and pure states over $\Ss^{\otimes G}$ will be denoted by ${\rm S}(\Ss^{\otimes G})$ and ${\rm PS}(\Ss^{\otimes G})$, respectively. They also have natural $G$-actions.

There is much interest in time evolutions on $\Ss^{\otimes G}$ generated by $G$-invariant inner-limit $\ast$-derivations $\delta_{\bm h}$, $\bm h =\{h_\Lambda\}$, induced from nets of inner derivations (see \cite{BratteliBook1,BratteliBook2} and section~\ref{Sec:FFDer})
\begin{equation}\label{Eq:FFPDer0}
 \Ss^{\otimes G} \ni a \mapsto  \delta_\Lambda(a)={\rm i}(h_\Lambda a -a h_\Lambda), \quad h_\Lambda = \sum_{g\cdot \Delta \subseteq \Lambda} \alpha_g(q) \in \Ss^{\otimes G}.
\end{equation}
Here, $\Delta$ is a fixed finite subset of $G$ called the range, $q \in \Ss_\Delta$ is called a finite-range interaction, and $\Lambda$ samples the countable set ${\rm K}(G)$. $\{h_\Lambda\}$ is said to be frustration-free if $q$ is a positive element, yet $0$ belongs to the spectrum ${\rm Spec}(h_\Lambda)$ for all $\Lambda$'s. The models with frustration-free interactions played and continue to play an important role in the process of understanding the dynamics of correlated quantum systems \cite{NachtergaeleLMP2024} and in the classification program of gapped quantum systems \cite{OgataICM2023}. Their characteristics make them amenable to specialized theoretical and numerical techniques and, as such, much insight has been gained about these models. For example, they are among the models whose ground states can be decided to be gapped by using numerically implementable algorithms \cite{NachtergaeleCMP1996}.

This is a very special situation. Indeed, the family of supporting projections $p_\Lambda$ of such $h_\Lambda$'s display several remarkable properties: (1) $p_\Lambda$ is proper and localized in $\Ss_\Lambda$; (2) $p_{\Lambda'}\leq p_\Lambda $ ($\Leftrightarrow p_{\Lambda'}^\bot \geq p_\Lambda^\bot$), whenever $\Lambda' \subseteq \Lambda$; (3) $\alpha_g(p_\Lambda)=p_{g \cdot \Lambda}$. These mentioned properties make the family $\{p_\Lambda\}$ into what we call a frustration-free proper $G$-system of projections (see definition~\ref{Def:FFSys}). We denote their set by $\mathfrak F_G(\Ss^{\otimes G})$ and we also consider the larger set $\mathfrak F(\Ss^{\otimes G})$ of families of proper projections displaying only properties (1) and (2) listed above. Any $\{p_\Lambda\} \in \mathfrak F_G(\Ss^{\otimes G})$ generates a large family of frustration-free inner-limit $G$-derivations, such as
\begin{equation}\label{Eq:FFPDer}
h_\Lambda (p_\Delta)= \sum_{g \cdot \Delta \subseteq \Lambda} \alpha_g(p_\Delta)=\sum_{g \cdot \Delta \subseteq \Lambda} p_{g \cdot \Delta}, \quad\Lambda,\Delta\in {\rm K}(G).
\end{equation}
In fact, we can replace $p_\Delta$ by any $q =q^\ast \in \Ss_\Delta$ with $q \leq c \, p_\Delta$ for some $c\in \RM_+$, because then $\sum_{g \cdot \Delta \subseteq \Lambda} \alpha_g(q)\leq c|\Lambda|p_\Lambda$ and $0$ is necessarily in the spectrum of $\sum_{g \cdot \Delta \subseteq \Lambda} \alpha_g(q)$, because $p_\Lambda$'s are proper. Note that the second expression in \eqref{Eq:FFPDer} can be used in cases where the $G$-equivariance is absent. For these reasons, one can shift the focus from models to frustration-free systems of projections. As we shall see, $\mathfrak F(\Ss^{\otimes G})$ can be endowed with a natural partial order such that it becomes a $\wedge$-semilattice compatible with the $G$-action. 

Related to frustration-free models, the class of frustration-free ground states were introduced in \cite{AffleckCMP1988}. Given a frustration-free model as in \eqref{Eq:FFPDer0}, $\omega \in {\rm S}(\Ss^{\otimes G})$ is a frustration-free ground state for $\delta_{\bm h}$ if $\omega (p_\Lambda)=0$ for all $\Lambda \in {\rm K}(G)$. One will like to identify the states of $\Ss^{\otimes G}$ that are frustration-free ground states for at least one model, which we will simply call frustration-free states. We answer this question by appealing to the space $\hat{\mathfrak H}(\Ss^{\otimes G})$ of hereditary $C^\ast$-subalgebras of $\Ss^{\otimes G}$, briefly reviewed in section \ref{Sec:Background}. It is relevant here because, first, any state $\omega$ induces a canonical hereditary $C^\ast$-subalgebra $\Bb_\omega=\Nn_\omega^\ast \cap \Nn_\omega$, where $\Nn_\omega$ is the left-ideal appearing in the GNS construction, and, secondly, there is an injective map from $\hat{\mathfrak H}(\Ss^{\otimes G})$ to the subsets of ${\rm PS}(\Ss^{\otimes G})$, with a natural inverse on its range. We call the subset $\Omega_\omega \subset {\rm PS}(\Ss^{\otimes G})$ corresponding to $\Bb_\omega$ the (hereditary) footprint of $\omega$. Now, let $\mathfrak H^F(\Ss^{\otimes G})$ be the class of hereditary $C^\ast$-subalgebras displaying the property 
\begin{equation}\label{Eq:PropF1}
       {\rm (F):} \ \  \Bb = \overline{\cup_{\Lambda \in {\rm K}(G)} (\Bb \cap \Ss_\Lambda)}.
\end{equation}
Then, if $\Bb_\omega$ has property (F), $\omega$ is a frustration-free ground state for at least one frustration-free derivation, which can be constructively derived from $\Bb_\omega$. In fact, a state is frustration-free {\it if and only if} its footprint is contained by the footprint of a state displaying property (F) (see section \ref{Sec:FFvsStates}). Therefore, property (F) completely and intrinsically characterizes the frustration-free states over $\Ss^{\otimes G}$. 
 
Furthermore, we will construct a split epimorphism from $\mathfrak F(\Ss^{\otimes G})$ to the poset $\mathfrak H^F(\Ss^{\otimes G})$ of hereditary $C^\ast$-subalgebras satisfying property (F) (see section \ref{Sec:FFfromB}), which can be used to answer another question, namely, how restrictive are the frustration-free models? For this, we exploit the bijective relation between $\mathfrak H(\Ss^{\otimes G})$ and the space $\Pp_o\big ((\Ss^{\otimes G})^{\ast \ast}\big)$ of open projections in the double dual $(\Ss^{\otimes G})^{\ast \ast}$. Specifically, any hereditary $C^\ast$-subalgebra comes in the form $p \Ss^{\otimes G} p \cap \Ss^{\otimes G}$ for an open projection $p$ from the double dual.\footnote{The standard is to use $p(\Ss^{\otimes G})^{\ast \ast} p \cap \Ss^{\otimes G}$ for the hereditary $C^\ast$-subalgebra, but our writing is equivalent (see Eq.~1.1 in \cite{BosaMPCPS2018}).} A ground state $\omega$ of a generic quantum spin model selects the unique open projection corresponding to $\Bb_\omega$, and property (F) supplies an intrinsic characterization of the space $\Pp_o^F\big ((\Ss^{\otimes G})^{\ast \ast}\big)$ of those open projections generated as weak$^\ast$-limits of families of frustration-free projections. Now, the double dual is naturally a $W^\ast$-algebra that can be concretely realized as the weak closure of the universal representation $\pi_u = \bigoplus_{\eta \in {\rm S}(\Ss^{\otimes G}) } \pi_\eta$ of $\Ss^{\otimes G}$.\footnote{Throughout, for $\eta$ a state, $\pi_\eta$ denotes its GNS representation.} We will show that $\Pp_o^F\big ((\Ss^{\otimes G})^{\ast \ast}\big)$ is dense in $\Pp_o\big ((\Ss^{\otimes G})^{\ast \ast}\big)$ in the norm topology of $\BM(\Hh_u)$ (see section \ref{Sec:FFOpenProj}). This can be seen as a sharpening and generalization of the known results for the case $G=\ZM$ \cite{FannesCMP1992}.

In fact, the connection to the open projections in the double dual supplies new tools to investigate the frustration-free ground states. In the literature, one can find a popular condition on a frustration-free system of projections, known as the local topological quantum order (LTQO) condition (see \ref{Prop:LTQO}), that assures that its set of frustration-free ground states consists of one point. Using the tools we just mentioned, we identify in section \ref{Sec:OptLTQO} the optimal version of LTQO that gives the sufficient and necessary condition for the set of the frustration-free ground states to consists of one point.

Thus, there are several problems related to frustration-free models whose solutions follow from the relations summarized in the following diagram:
\begin{equation}\label{Eq:MasterDiag}
\begin{tikzcd}
    & \Pp_o^F\bigl((\Ss^{\otimes G})^{**}\bigr) & \\
    \mathfrak F(\Ss^{\otimes G})
    &&
    \mathfrak H^F(\Ss^{\otimes G})
    \\
    & \mathfrak P^F\bigl({\rm PS}(\Ss^{\otimes G})\bigr) &
    \arrow[from=2-1, to=1-2]
    \arrow[from=1-2, to=2-3]
    \arrow[from=2-3, to=1-2]
    \arrow["{\mathfrak p}^{-1}", bend left=12, from=2-1, to=2-3]
    \arrow["{\mathfrak p}", hook, bend left=12, from=2-3, to=2-1]
    \arrow[from=3-2, to=2-3]
    \arrow[from=2-3, to=3-2]
\end{tikzcd}
\end{equation}
This diagram records the
existence and direction of the relevant maps. In particular,
$\mathfrak p^{-1}$ is only the left inverse of $\mathfrak p$.
The main goal of our paper is to describe the structure of its entries, establish the seen relations, and demonstrate a number of direct consequences. As we shall see, all the entries are posets with partial orders that are compatible with the $G$-actions, and all the maps seen in \eqref{Eq:MasterDiag} are $G$-equivariant morphisms of posets. As such, all our statements can be specialized for $G$-invariant models, states, etc. In our opinion, property (F) and the diagram \eqref{Eq:MasterDiag} encode the essence of frustration-freeness, because now the concept can be defined for any AF-algebra, once a preferential finite-dimensional filtration is chosen.

\medskip

\noindent
{\bf Acknowledgements:}  This work was supported by the U.S. National Science Foundation through the grant CMMI-2131760, and by U.S. Army Research Office through contract W911NF-23-1-0127. The authors acknowledge fruitful discussions during the workshop "Quantum Field Theory and Topological Phases via Homotopy Theory and Operator Algebras" organized by CMSA at Harvard University, especially with Corey Jones, David Penneys, and Xiao-Gang Wen.

\section{Frustration-free inner-limit derivations and systems of projections}\label{Sec:FFDer}

This section recalls the definition of frustration-free derivations and introduces the related concept of frustration-free systems of projections. Basic properties of the latter are established and examples are supplied. We will make use of the notation already introduced in the introduction.

Our main results are specialized for the spin algebras $\Ss^{\otimes G}$ over the Cayley graphs of a finitely generated infinite groups. The site algebra is $\Ss = M_d(\CM)\simeq \BM(\Hh)$, $\Hh = \CM^d$.  An element $a\in \Ss^{\otimes G}$ is called a quasi-local physical observable. If $a\in \Ss_\Lambda$ for some $\Lambda\in {\rm K}(G)$, one calls $a$ a local observable and it is standard to denote by ${\rm supp}(a)$ as the smallest possible $\Lambda$ such that $a\in \Ss_\Lambda,$ and refer to it as the support of $a.$ 

We introduce now the class of derivations of interest to us. Let $q$ be a self-adjoint element from one of the finite dimensional subalgebras $\Ss_\Delta$, and consider the net $\bm h =\{h_\Lambda\}_{\Lambda \in {\rm K}(G)}$ of elements
    \begin{equation}\label{Eq:HLambda}
        h_\Lambda(q):= \sum_{g \cdot \Delta \subseteq \Lambda}\alpha_g(q-\epsilon_\Lambda (q)\, 1) \in \Ss_\Lambda \subset \Ss^{\otimes G},
    \end{equation}
    where $\epsilon_\Lambda(q)$'s are real parameters, entirely determined by $q$ and $\Lambda$, enforcing the spectral condition\footnote{This convention removes the freedom of adding multiples of identity.}
    \begin{equation}\label{Eq:SpecCond}
        \min \, {\rm Spec} (h_\Lambda) =0.
    \end{equation} 
    Let $\Dd_{\bm h}:= \bigcup_{\Lambda \in {\rm K}(G)} \, \Ss_{\Lambda}$, which is a dense and $G$-invariant $\ast$-subalgebra of $\Ss^{\otimes G}$. Then
   \begin{equation}
       {\rm K}(G) \ni \Lambda \mapsto [h_\Lambda,a] := h_\Lambda a - a h_\Lambda \in \Ss^{\otimes G}
   \end{equation} 
   is a Cauchy net for all $a \in \Dd_{\bm h}$, and we can define a derivation on $\Dd_{\bm h}$ 
   \begin{equation}
       \delta_{\bm h}(a) : = \lim {\rm i}[h_\Lambda,a].
   \end{equation}
   We call such $\delta_{\bm h}: \Dd_{\bm h} \to \Ss^{\otimes G}$ inner-limit $G$-derivations,\footnote{We prefer the label inner-limit introduced in \cite{BratteliCMP1975} over the label almost-inner.} because their domains and actions are $G$-invariant. Additionally, $\delta_{\bm h}(a^\ast)=\delta_{\bm h}(a)^\ast$ for all $a \in \Dd_{\bm h}$, hence they are also $\ast$-derivations. It is known that such inner-limit derivations are closable and that the closures generate dynamics on $\Ss^{\otimes G}$ \cite[Example~3.2.25]{BratteliBook2}.

\begin{definition}[\cite{NachtergaeleLMP2024}]\label{Def:FFDer}
An inner-limit $G$-derivation of the type~\eqref{Eq:HLambda} is said to be frustration-free if 
\begin{equation}
\epsilon_\Lambda(q) = {\rm min}\, {\rm Spec}(q)
\end{equation}
for all $\Lambda$'s.
\end{definition}

The spectra of $h_\Lambda$'s from~\ref{Eq:HLambda} are finite, hence we can speak of supporting projections $p_\Lambda \in \Ss_\Lambda$, {\it i.e.} of the spectral projections onto ${\rm Spec}\, h_\Lambda \setminus \{0\}$. These $p_\Lambda$'s display frustration-freeness \cite{NachtergaeleLMP2024}
\begin{equation}\label{Eq:FF1}
        p_{\Lambda_1}^\bot \, p_{\Lambda_2}^\bot =  p_{\Lambda_2}^\bot \, p_{\Lambda_1}^\bot = p_{\Lambda_2}^\bot 
\end{equation}
or, equivalently,
\begin{equation}\label{Eq:FF2}
        p_{\Lambda_1}\, p_{\Lambda_2} =  p_{\Lambda_2}\, p_{\Lambda_1}= p_{\Lambda_1}
\end{equation}
for any pair $\Lambda_1 \subseteq \Lambda_2 \in {\rm K}(G)$. This prompts the following definition:

\begin{definition}\label{Def:FFSys}
    A family of projections
    \begin{equation}
        \big \{p_\Lambda \in \Ss_\Lambda \subset \Ss^{\otimes G} \big \}_{\Lambda \in {\rm K}(G)}
    \end{equation}
    is called a frustration-free system if the equivalent statements~\eqref{Eq:FF1} and \eqref{Eq:FF2} apply and $p_\Lambda\neq 0$ for some $\Lambda\in {\rm K}(G)$. If in addition 
    \begin{equation}\label{Eq:PGEq}
    \alpha_g(p_\Lambda)=p_{g \cdot \Lambda}, \quad \forall \ \Lambda \in {\rm K}(G),
    \end{equation}
    we call $\{p_\Lambda\}$ a $G$-system. The system will be called proper if $p_\Lambda < 1_\Lambda$ for all $\Lambda \in {\rm K}(G)$.
\end{definition}

\begin{remark}
    {\rm If $p_\Lambda =1_\Lambda$ for some $\Lambda \in {\rm K}(G)$, then, for any $\Lambda' \supseteq \Lambda$, $p_\Lambda$ is embedded in $\Ss_{\Lambda'}$ as $1_{\Lambda'}$ and the frustration-free condition forces $p_{\Lambda'} =1_{\Lambda'}$, rendering the system trivial. As such, our focus will be primarily on the proper systems for which, as the definition above says, $p_\Lambda < 1_\Lambda$ for {\it all} $\Lambda\in {\rm K}(G)$.}
\end{remark}

\begin{remark}
    {\rm The set ${\mathfrak F}(\Ss^{\otimes G})$ of frustration-free proper systems of projections accepts the following natural $G$-action:
    \begin{equation}
        g \cdot \{p_\Lambda \in \Ss_\Lambda\}_{\Lambda \in {\rm K}(G)} := \{\alpha_g (p_{g^{-1} \cdot \Lambda}) \in \Ss_\Lambda\}_{\Lambda \in {\rm K}(G)}.
    \end{equation}
    Then~\eqref{Eq:PGEq} can be restated as $g\cdot \{p_\Lambda\}=\{p_\Lambda\}$, and the set ${\mathfrak F}_G(\Ss^{\otimes G})$ of proper $G$-systems is just the subset of ${\mathfrak F}(\Ss^{\otimes G})$ of elements fixed by this action. The same observations apply to the set $\hat{\mathfrak F}(\Ss^{\otimes G})$ of not necessarily proper frustration-free systems of projections.}$\Diamond$
\end{remark}

Every $C^\ast$-algebra comes equipped with a standard partial ordering. Thus, the set $\Pp(\Ss^{\otimes G})$ of projections is a poset and $p \leq q$ if and only if $p q=q p = p$ \cite[Lemma~14.2.2.]{StrungBook}. However, $\Pp(\Ss^{\otimes G})$ is not a lattice. Indeed, one can read from the Bratteli diagram of $\Ss^{\otimes G}$ that this $C^\ast$-algebra is not postliminal \cite{LazarTAMS1980} and, for $\Pp(\Ss^{\otimes G})$ to be a lattice, this is a necessary condition \cite{LazarMS1982}. In fact, the set of projections of AF algebras rarely form a lattice under the standard order \cite{LazarTAMS1983}. However, since $\Ss_\Lambda \simeq \BM(\Hh_\Lambda)$, $\Hh_\Lambda:=\Hh^{\otimes \Lambda}$, then the projections of $\Ss_\Lambda$ do form a lattice and:\footnote{This detail makes $\Ss^{\otimes G}$ special and, as such, many stated properties will not be applicable to general AF algebras of physical observables, unless we remove the references to lattice structures from them.}

\begin{proposition}\label{Prop:FLattice}
   $\hat{\mathfrak F}(\Ss^{\otimes G})$ is a lattice when gifted with the partial order 
\begin{equation}\label{Eq:FFOrder}
    \{p_\Lambda^1\} \leq \{p_\Lambda^2\} \ \Leftrightarrow \ p_\Lambda^1 \leq p_\Lambda^2 \ \ \forall \ \Lambda \in {\rm K}(G).
\end{equation}
In that case,
\begin{equation}\label{Eq:FWedge}
    \{p_\Lambda^1\} \wedge \{p_\Lambda^2\} = \{p_\Lambda^1 \wedge p_\Lambda^2\}
\end{equation}   
and 
\begin{equation}\label{Eq:FVee}
    \{p_\Lambda^1\} \vee \{p_\Lambda^2\} = \{p_\Lambda^1 \vee p_\Lambda^2\}.
\end{equation}
\end{proposition}

\begin{proof}
    Clearly, \eqref{Eq:FFOrder} is a partial order. To see if it admits the suprema and infima as stated, we first need to verify that the right sides of~\eqref{Eq:FWedge} and \eqref{Eq:FVee} are frustration-free systems. Seeing $\Ss_\Lambda$ as $\BM(\Hh_\Lambda)$, the frustration-freeness of $\{p_\Lambda^i\}$, for $i=1,2$, is equivalent to  ${\rm Ran} \, p_\Lambda^i  \subseteq {\rm Ran} \, p_{\Lambda'}^i \subseteq \Hh_{\Lambda'}$, whenever $\Lambda \subseteq \Lambda'$. Then
    \begin{equation}
        {\rm Ran} \, p_\Lambda^1 \cap {\rm Ran}\, p_\Lambda^2 \subseteq {\rm Ran} \, p_{\Lambda'}^1 \cap {\rm Ran}\, p_{\Lambda'}^2
    \end{equation}
    and 
    \begin{equation}
        {\rm Ran} \, p_\Lambda^1 + {\rm Ran}\, p_\Lambda^2 \subseteq {\rm Ran} \, p_{\Lambda'}^1 + {\rm Ran}\, p_{\Lambda'}^2,
    \end{equation}
    which confirm that $ \{p_\Lambda^1 \wedge p_\Lambda^2\}$ and $ \{p_\Lambda^1 \vee p_\Lambda^2\}$ are indeed frustration-free systems. Now, if $\{p_\Lambda\} \leq \{p_\Lambda^i\}$, $i=1,2$, then necessarily $p_\Lambda \leq p_\Lambda^1 \wedge p_\Lambda^2$ for all $\Lambda \in {\rm K}(G)$ and, as such, 
    \begin{equation}
    \{p_\Lambda\} \leq \{p_\Lambda^1 \wedge p_\Lambda^2\}.
    \end{equation}
    Similarly, if $\{p_\Lambda\} \geq \{p_\Lambda^i\}$, $i=1,2$, then necessarily $p_\Lambda \geq p_\Lambda^1 \vee p_\Lambda^2$ for all $\Lambda \in {\rm K}(G)$ and, as such, 
    \begin{equation}
        \{p_\Lambda\} \geq \{p_\Lambda^1 \vee p_\Lambda^2\}.
    \end{equation}
    These show that infima and suprema exist and are given by \eqref{Eq:FWedge} and \eqref{Eq:FVee}, respectively.
\end{proof}

\begin{remark}
    {\rm The lattice structure identified above is compatible with the $G$-action. As a consequence, $\hat{\mathfrak F}_G(\Ss^{\otimes G})$ is a sub-lattice. Of course, the subsets ${\mathfrak F}(\Ss^{\otimes G})$ and ${\mathfrak F}_G(\Ss^{\otimes G})$ of proper systems are only $\wedge$-semilattices.} $\Diamond$
\end{remark}

The following example assures us that the set of frustration-free $G$-systems is not void for any finitely generated torsion-free discrete group.

\begin{figure}
    \centering
    \includegraphics[width=0.5\linewidth]{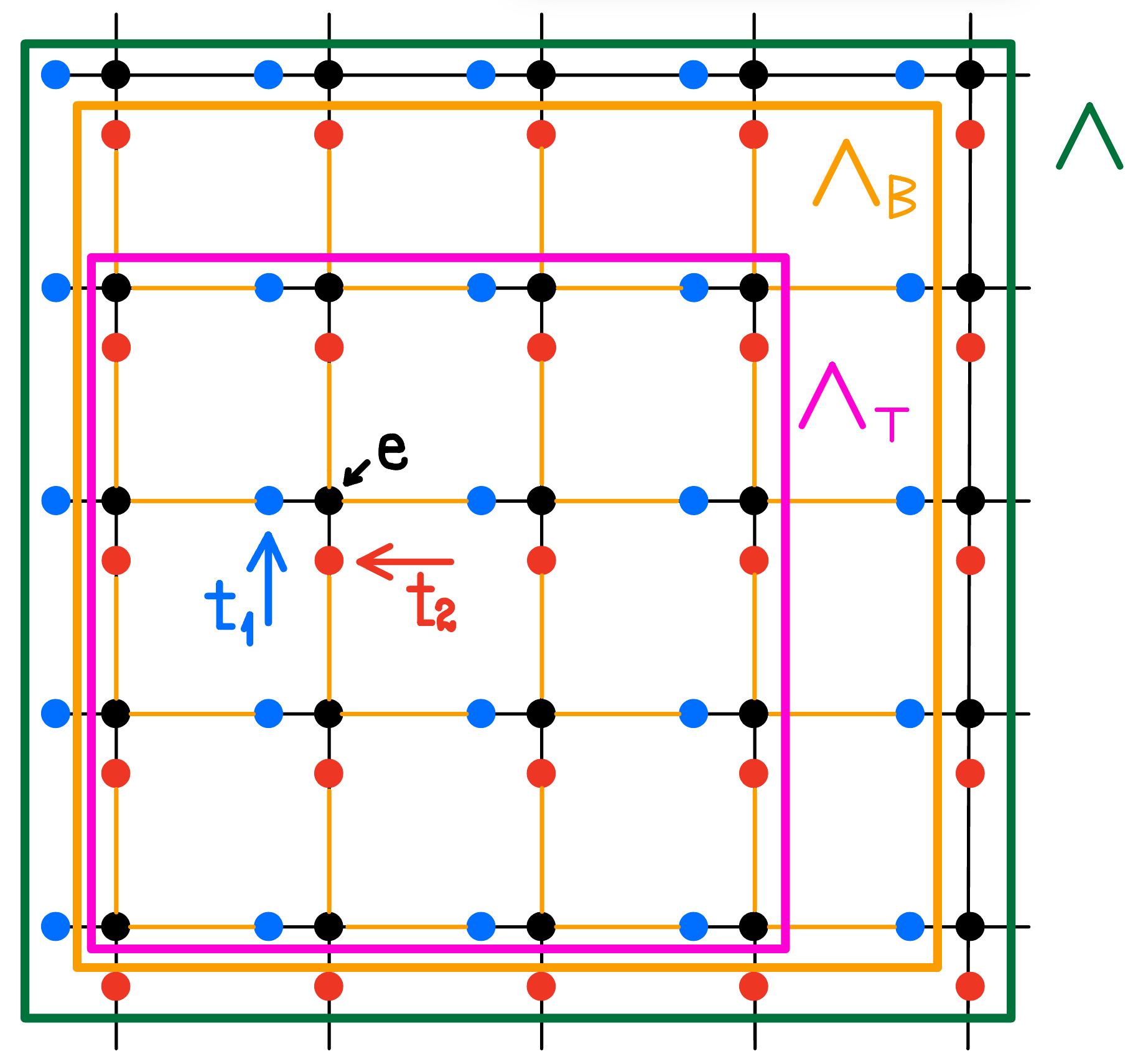}
    \caption{Cayley graph and its decoration for the valence-bond construction in the case $G=\ZM^2$. The vertices of the decorated graph connected by orange bonds form the subset $\Lambda_B$. They support the indices called by $\Gamma$'s in Eq.~\eqref{Eq:Weights}. The vertices of the undecorated graph forming the subset $\Lambda_T$ are also shown.}
    \label{Fig:VB1}
\end{figure}

\begin{example}\label{Ex:VBS}
    {\rm The construction is a variant of valence bond states \cite{AffleckCMP1988,NachtergaeleATMP2021} or more general PEP-states \cite{JahromiPRB2019}. Consider a torsion-free finitely generated group $G$ and its tensor $C^\ast$-algebra $\Ss^{\otimes G}$ with $\Ss = M_d(\CM)$. Let $\{t_1,\ldots, t_k\}$ be a minimal set of generators of $G$, giving us a geometric picture of $G$ in the form of a Cayley graph. Then we let $T=\{e,t_1,\ldots,t_k\}$, where $e$ is the neutral element of $G$. Each element $t \in T$ carries a finite set of indices $J(t)$, and we fix a map 
    \begin{equation}     
    \prod_{t \in T} J(t) \ni \{x(t)\}:= x \mapsto \psi(x) \in \Hh := \CM^d.
      \end{equation}
    We also fix a set of coefficients via the map
         \begin{equation}    
         \prod_{t \in T} J(t) \ni x \mapsto \Gamma(x) \in \CM, \quad \sum |\Gamma(x)|>0.
      \end{equation} 
   Additionally, we consider the decorated graph $G \times T$ and a point $(g,t)$ of this graph will carry the set of indices $J(t)$. See Fig.~\ref{Fig:VB1} for an illustration when $G=\ZM^2$ and $t_i$ are its standard generators. Also, for orientation, all the above play out for AKLT model \cite{AffleckCMP1988} as follows: $G=\ZM$, $\Ss=M_3(\CM)$, $T=\{e,t\}$ and $J(e)=J(t)=\{1,2\}$. Furthermore, the $\psi$ map is defined in equation (2.3) of \cite{AffleckCMP1988} and $\Gamma(\alpha,\beta)=\epsilon^{\alpha\beta}$, the anti-symmetric tensor.

A configuration of this decorated Cayley graph equates with a choice 
\begin{equation}
    \{x_g(t)\}_{g \in G}^{t \in T}=\{x_g\}_{g \in G}
\end{equation}
of indices carried by the decorated lattice. By restriction, we can define local configurations, and we denote by $\Cc(\Xi)$ the set of configurations over $\Xi \in {\rm K}(G \times T)$. Note that if $\Xi$ and $\Xi'$ are disjoint finite subsets of $G \times T$, then $X_\Xi =\{x_g(t)\}_{(g,t) \in \Xi}$ and $X'_{\Xi'} =\{x'_g(t)\}_{(g,t) \in \Xi'}$ can be joined into a configuration $X_\Xi \curlyvee X'_{\Xi'}$ over $\Xi \cup \Xi'$ and we have $\Cc(\Xi) \curlyvee \Cc(\Xi') = \Cc(\Xi \cup \Xi')$. Furthermore, if we consider the canonical maps
\begin{equation}
    X_{\Lambda \times T} = \{x_g\}_{g \in \Lambda} \mapsto \Psi(X_{\Lambda \times T}): = \bigotimes_{g \in \Lambda} \, \psi(x_g) \in \Hh_\Lambda
\end{equation}
for each $\Lambda \in {\rm K}(G)$, then 
\begin{equation}\label{Eq:Psi2}
    X_{\Lambda \times T} \curlyvee X'_{\Lambda' \times T} \mapsto \Psi(X_{\Lambda \times T}) \otimes \Psi(X'_{\Lambda' \times T})
\end{equation}
 if $\Lambda$ and $\Lambda'$ are disjoint. Note that, if $\Xi$ and $\Xi'$ are disjoint strict subsets of $\Lambda \times T$, then $\Psi(X_\Xi \curlyvee X_{\Xi'})$ makes sense if $\Xi \curlyvee \Xi' = \Lambda \times T$, but that is not the case for $\Psi(X_\Xi)$ or $\Psi(X_{\Xi'})$. In particular, \eqref{Eq:Psi2} does not apply for such cases.
 
 Now, for each local configuration $X_{\Lambda \times T}$ over $\Lambda \times T$, we define an amplitude in the form
\begin{equation}\label{Eq:Weights}
    C_\Lambda(X_{\Lambda \times T}) : = \prod_{g\in \Lambda_T} \Gamma(T \triangleright x_g),
\end{equation}
where $\Lambda_T : = \bigcap_{t \in T}(t^{-1} \Lambda) \in {\rm K}(G)$ and 
\begin{equation}  
    T \triangleright x_g = T \triangleright \{x_g(t)\}:= \{x_{t g}(t)\}\in \prod_{t \in T} J(t).
\end{equation}
It is useful to define yet another set, namely, the set given by
\begin{equation}
    \Lambda_B:=\{(tg,t)\mid g\in \Lambda_T,\;t\in T\}\subset \Lambda\times T,
\end{equation}
containing all the sites of the decorated lattice that support the indices $T \triangleright x_g$ for $g \in \Lambda_T$.
Intuitively, these are the sites of the decorated lattice connected or bonded by $\Gamma$'s in Eq.~\eqref{Eq:Weights} (see Fig.~\ref{Fig:VB1}). Note that $\Lambda_B$ is a subset of the decorated lattice  and not of the lattice itself, that $\Lambda_T \times \{e\} \subset \Lambda_B$, and that $C_\Lambda$ in \eqref{Eq:Weights} depends only on the restriction of $X_{\Lambda \times T}$ to $\Lambda_B$. Lastly, for each $\Lambda \in {\rm K}(G)$ with $\Lambda_T \neq \emptyset$, we define the linear subspace
    \begin{equation}
\Vv_\Lambda : = {\rm Span}\Big \{\sum_{ X\in \Cc(\Lambda_B)} C_\Lambda(X) \, \Psi(Y \curlyvee X), \ Y \in \Cc(\Lambda \times T \setminus \Lambda_B) \Big \}\subseteq \Hh_\Lambda.
\end{equation}
If $\Lambda_T = \emptyset$, we set $\Vv_\Lambda = \Hh_\Lambda$. Note that, if $\bar J$ is the largest among the cardinals of $J(t)$, then the number of $Y$ configurations is smaller than $\bar J^{|T||\Lambda \setminus \Lambda_T|}$. Since ${\rm dim}\, \Hh_\Lambda = d^{|\Lambda|}$, we see that ${\rm dim}\, \Vv_\Lambda \leq {\rm dim}\, \Hh_\Lambda$ for sufficiently large $\Lambda$.

\begin{proposition}
    For a linear Hilbert subspace $\Ww \in \Hh_\Lambda$, we denote its orthogonal complement by $\Ww^\perp$ and we let $\langle \Ww \rangle\in \BM(\Hh_\Lambda) \simeq \Ss_\Lambda$ be the associated orthogonal projection. Then $\{p_\Lambda : =\langle \Vv_\Lambda^\bot \rangle\}$ is a frustration-free system of projections. Note that $p_\Lambda^\perp= \langle \Vv_\Lambda \rangle$. Furthermore, if $\Vv_\Lambda \neq \{0\}$ for all $\Lambda \in {\rm K}(G)$, then the system is proper.
\end{proposition}
\begin{proof}
    The prior analysis assures us that $p_\Lambda \neq 0$ for any large enough $\Lambda$. Now take $\Xi \supseteq \Lambda \in {\rm K}(G)$. If $\Lambda_T = \emptyset$, then automatically $p_\Xi \geq p_\Lambda = 0$, so let us assume that $\Lambda_T \neq \emptyset$. We will show that $\Vv_\Xi \subseteq \Hh_{\Lambda^c} \otimes \Vv_\Lambda$, which is equivalent to $p_\Xi^\bot \leq p_\Lambda^\bot$. Here and throughout, $\Lambda^c = \Xi\setminus \Lambda$. To proceed, we recall that the generating vectors of $\Vv_\Xi$ take the form
\begin{equation}\label{Eq:VXi}
 \sum_{X \in \Cc(\Xi_B)} C_\Xi(X) \Psi(Y \curlyvee X), \quad Y \in \Cc(\Xi \times T \setminus \Xi_B), 
\end{equation}
and the key property of the construction is the factorization
\begin{equation}
    C_\Xi(X)=C_\Xi(X_{\Xi_B \setminus \Lambda_B} \curlyvee X_{\Lambda_B})= F(X_{\Xi_B \setminus \Lambda_B}) \, C_\Lambda (X_{\Lambda_B}).
\end{equation}
Indeed,
\begin{equation}
 C_\Xi(X) = \prod_{g\in \Xi_T} \Gamma(T \triangleright x_g) = \Big [\prod_{g\in \Xi_T\setminus \Lambda_T} \Gamma(T \triangleright x_g) \Big ] \Big[ \prod_{g\in \Lambda_T} \Gamma(T \triangleright x_g) \Big ]
\end{equation}
and the first factor of the right side does not involve any of the indices supported by $\Lambda_B$. Our next task is to factorize $\Psi(Y \curlyvee X)$. If $\Gamma_B := \Xi_B \cap (\Lambda \times T)$ and $\Gamma_B^c$ is its complement inside $\Lambda \times T$, then $\Xi_B \supseteq \Gamma_B \supseteq \Lambda_B$ and, as such, we have the disjoint decomposition 
\begin{equation}
   \Xi_B \setminus \Lambda_B = (\Xi_B \setminus \Gamma_B) \cup (\Gamma_B \setminus \Lambda_B).
\end{equation}
Furthermore
\begin{equation}
    \Xi\times T \setminus \Xi_B = (\Lambda^c \times T \setminus \Xi_B)  \cup \Gamma_B^c.
\end{equation}
% \begin{equation}
%\begin{aligned}
%    & \Xi\times T \setminus \Xi_B = ( (\Xi\setminus\Lambda) \times T \cup \Lambda \times T) \setminus \Xi_B \\
%&= [(\Xi\setminus\Lambda) \times T \setminus \Xi_b] \cup (\Lambda \times T \setminus \Xi_B) \\
%& = [(\Xi\setminus\Lambda) \times T \setminus \Xi_b] \cup [(\Lambda\times T) \setminus (\Xi_B \cap (\Lambda \times T)]
%    \end{aligned}
%\end{equation}
The latter enables us to decompose the configuration $Y$ as $Y= Q \curlyvee Y_{\Gamma_B^c}$, where $Q$ is a configuration over $\Lambda^c \times T \setminus \Xi_B$. Thus, we can factorize the configurations and the weights in \eqref{Eq:VXi} as
\begin{equation}
\begin{aligned}
  \sum_{V \in \Cc(\Gamma_B \setminus \Lambda_B)} \, \sum_{Z \in \Cc(\Xi_B \setminus \Gamma_B)} & F(Z \curlyvee V) \Psi(Q  \curlyvee Z)  \\
 & \qquad \otimes \sum_{W\in \Cc( \Lambda_B)} C_\Lambda(W) \Psi(Y_{\Gamma^c_B} \curlyvee V \curlyvee W).
 \end{aligned}
\end{equation}
The factorization is sound because $Q \curlyvee Z$ is a configuration over 
\begin{equation}
    (\Lambda^c \times T \setminus \Xi_B) \cup (\Xi_B \setminus \Gamma_B) = \Lambda^c \times T
\end{equation}
and $Y_{\Gamma^c_B} \curlyvee V \curlyvee W$ is a configuration over
\begin{equation}
    \Gamma_B^c \cup (\Gamma_B \setminus \Lambda_B) \cup \Lambda_B = \Lambda \times T.
\end{equation}
This can be processed to
\begin{equation}
 \sum_{V \in \Cc(\Gamma_B \setminus \Lambda_B)} \Phi(V) \otimes \sum_{W \in \Cc( \Lambda_B)} C_\Lambda(W) \Psi(Y_{\Gamma^c_B} \curlyvee V \curlyvee W),
\end{equation}
where all $\Phi(V)$ reside in $\Hh_{\Lambda^c}$ and $Y_{\Gamma^c_B} \curlyvee V \in \Lambda \times T \setminus \Lambda_B$.
\end{proof}

Regarding the condition $\Vv_\Lambda \neq 0$, let us point out that, if it occurs for one subset, then it occurs also for all larger subsets. This imposes very special constraints on the $\Gamma$-coefficients and $\psi$ function, which happen with zero probability if this input data is generated randomly. Lastly, the construction uses exclusively the left action of $G$ and, because the latter commutes with the right action used in the definition of the $G$-action on the systems of projections, $\{p_\Lambda\}$ is actually a proper $G$-system.
}$\Diamond$
\end{example}

If $\{p_\Lambda\}$ is a frustration-free system of proper projections, then, for any pair $\Lambda_1,\Lambda_2 \in {\rm K}(G)$, we have $ p_{\Lambda_1 \cup \Lambda_2} \geq p_{\Lambda_1} \vee p_{\Lambda_2}$ (or $ p_{\Lambda_1\cup \Lambda_2}^\bot \leq p_{\Lambda_1}^\bot \wedge p_{\Lambda_2}^\bot$). There are examples where equality actually takes place, at least for a subnet of $\{p_\Lambda\}$. Below we exhibit such example.

\begin{example}[\cite{FannesCMP1992}] 
 {\rm This example will also serve as introduction to yet another technique to produce frustration-free systems of projections. Take $G=\ZM$ and consider $\Ss = M_d(\CM) \simeq \BM(\CM^d)$. Let $\rho$ be an invertible $k\times k$ matrix and $\{v_\mu\}_{\mu=1}^d$ be a collection of $k \times k$ matrices such that 
       \begin{equation}\label{Eq:VV}
    \sum_\mu v_\mu v_\mu^\ast = 1_k, \quad \sum_\mu v_\mu^\ast \rho v_\mu = \rho
    \end{equation}
    and such that the linear map 
       \begin{equation}
           M_k(\CM) \ni B \mapsto \sum_\mu v_\mu B v_\mu^\ast \in M_k(\CM)
       \end{equation}
     has $1$ as a nondegenerate eigenvalue and all its other eigenvalues have modulus strictly smaller than $1$. Then, if $\{\psi_\mu\}$ is an orthonormal basis of $\CM^d$,
       \begin{equation}
         B \mapsto  \Gamma_n(B):= \sum_\mu \psi_{\mu_1} \otimes \ldots \otimes \psi_{\mu_n} \, {\rm Tr}\{ v_{\mu_n} \cdots v_{\mu_1}B\}
       \end{equation}
       is an injective map from $M_k(\CM)$ to $(\CM^d)^{\otimes n}$ for $n$ larger or equal than a threshold value $\ell \in \NM^\times$ \cite[Lemma 5.2]{FannesCMP1992}. Let $\Lambda_n := \{1, \ldots, n\}$ and define $p_{\Lambda_n}^\bot := \langle {\rm Ran} \, \Gamma_n \rangle \in \Ss_{\Lambda_n} \subset \Ss^{\otimes \ZM}$, as well as $p_{t\cdot \Lambda_n}:= \alpha_t(p_{\Lambda_n})$ for all $t\in \ZM$. This together with the remark at the end of the example supplies a $G$-system of projections. The following is in essence Lemma~5.5 from \cite{FannesCMP1992}:

\begin{proposition} Let $\Hh = \CM^d$. Then, for $m,n\leq p \in \NM^\times$,
    \begin{equation}\label{Eq:TT1}
      {\rm Ran} \, \Gamma_m \otimes \Hh^{\otimes (p-m)} \cap \Hh^{\otimes (p-n)} \otimes {\rm Ran}\, \Gamma_n = {\rm Ran}\, \Gamma_p,
    \end{equation}
whenever $p \leq m+n-\ell$. In other words, 
\begin{equation}\label{Eq:MP1}
    p_{\Lambda_m}^\bot \wedge p_{t\cdot \Lambda_n}^\bot = p_{\Lambda_m \cup t\cdot \Lambda_n}^\bot, \quad t=p-n. 
\end{equation}
\end{proposition}

\begin{proof}
    Let $\psi$ be a common vector of the spaces seen on the left side of~\eqref{Eq:TT1}. Then
\begin{equation}\label{Eq:Psi}
\begin{aligned}
\psi & =\sum_\mu \psi_{\mu_1} \otimes \cdots \otimes \psi_{\mu_p} {\rm Tr}\big[ v_{\mu_{p}} \cdots v_{\mu_{p-n+1}}C(\mu_{p-n},\ldots,\mu_1)] \\
& = \sum_\mu \psi_{\mu_1} \otimes \cdots \otimes \psi_{\mu_p} {\rm Tr}\big[D(\mu_p, \ldots,\mu_{m+1})v_{\mu_m} \cdots v_{\mu_1} \big],
\end{aligned}
\end{equation}
for families of $C$ and $D$ matrices from $M_k(\CM)$, indexed by the coefficients seen between the parentheses. As such, we must have
\begin{equation}\label{Eq:XX1}
\begin{aligned}
& \sum_\mu \psi_{\mu_1} \otimes \cdots \otimes \psi_{\mu_p} {\rm Tr}\big[v_{\mu_m} \cdots v_{\mu_{p-n+1}} \\
&\qquad \qquad \big(C(\mu_{p-n},\ldots,\mu_1)v_{\mu_p} \cdots v_{\mu_{m+1}} \\
& \qquad \qquad  \qquad \qquad  -v_{\mu_{p-n}} \cdots v_{\mu_1} D(\mu_p, \ldots,\mu_{m+1}) \big )\big]=0.
\end{aligned}
\end{equation}
Denoting the $k\times k$ matrices between the round parenthesis by $\tilde B$ and noticing that we have one such $\tilde B$ for each tuple $\{\mu_1,\ldots,\mu_{p-n},\mu_{m+1},\ldots,\mu_p\}$, the left side of Eq.~\eqref{Eq:XX1} becomes
\begin{equation}
\begin{aligned}
\sum_\mu \psi_{\mu_1} \otimes \cdots \otimes \psi_{\mu_{n-p}} \otimes & \Gamma_r [ \tilde B(\mu_1,\ldots,\mu_{p-n}, \\
& \qquad \qquad \mu_{m+1},\ldots,\mu_p) ] \otimes \psi_{\mu_{m+1}} \otimes \cdots \otimes \psi_{\mu_p },
\end{aligned}
\end{equation}
with $r=m-p+n$. In the stated conditions, the map $\Gamma_r$ is injective, hence Eq.~\eqref{Eq:XX1} holds if and only if all $\tilde B$'s are zero, or,
\begin{equation}
    C(\mu_{p-n},\ldots,\mu_1)v_{\mu_p} \cdots v_{\mu_{m+1}} = v_{\mu_{p-n}} \cdots v_{\mu_1} D(\mu_p, \ldots,\mu_{m+1})
\end{equation}
for all possible values of the seen indices. By repeatedly using~\eqref{Eq:VV}, we find
\begin{equation}
    C(\mu_{p-n},\ldots,\mu_1) = v_{\mu_{p-n}} \cdots v_{\mu_1} \sum_\mu D(\mu_p, \ldots,\mu_{m+1})v_{\mu_{m+1}}^\ast\cdots v_{\mu_p}^\ast ,
\end{equation}
and, by letting
\begin{equation}
    B = \sum_\mu D(\mu_p, \ldots,\mu_{m+1})v_{\mu_{m+1}}^\ast\cdots v_{\mu_p}^\ast ,
\end{equation}
$\psi$ in Eq.~\eqref{Eq:Psi} becomes $\Gamma_p(B)$ and the statement follows.
\end{proof}

We end the example by pointing out that this construction reduces to the one presented in \ref{Ex:VBS} if we take $J(0)=J(1)=\{1,\ldots,k\}$, $\psi(i,j)=\sum_\mu v_\mu(i,j)\psi_\mu$, and $\Gamma(i,j)=\delta_{ij}$ for $i,j=1, \ldots, k$.
   }$\Diamond$
\end{example}

\section{States, one sided ideals and hereditary subalgebras}\label{Sec:Background}

This section collects known facts in preparation for our analysis of the bottom half of the diagram \eqref{Eq:MasterDiag}. We will reserve a separate section for the top half of \eqref{Eq:MasterDiag}. Let us start by recalling that the set of states and the set of pure states over a $C^\ast$-algebra $\Aa$, ${\rm S}(\Aa)$ and ${\rm PS}(\Aa)$ respectively, inherit topologies from the dual $\Aa^\ast$ space equipped with the weak$^\ast$ topology and, if $\alpha$ is a $G$-action on $\Aa$, this action lifts to states $\alpha_g(\omega):= \omega \circ \alpha_g^{-1}$. Since pure states remain so after being acted on by an automorphism, both ${\rm S}(\Aa)$ and ${\rm PS}(\Aa)$ are dynamical $G$-systems, but we should recall that, while ${\rm S}(\Aa)$ is always compact for unital $C^\ast$-algebras, ${\rm PS}(\Aa)$ is not, in general. In fact, for $\Aa=\Ss^{\otimes G}$, ${\rm PS}(\Aa)$ is a weakly contractible and homogeneous space under the full group of automorphisms of $\Aa$ \cite{SpiegelJFA2025}, and ${\rm PS}(\Aa)$ is dense in ${\rm S}(\Aa)$ \cite[Example~4.1.31]{BratteliBook1}.

For $\omega \in {\rm S}(\Aa)$, we will use $\Nn_\omega$ to denote the closed left ideal
        \begin{equation}
        \Nn_\omega : = \{a \in \Aa \, | \, \omega(a^\ast a)=0 \}.
    \end{equation}
The left ideals of pure states are maximal and:

\begin{theorem}[Th.~5.3.5 \cite{MurphyBook}] If $\Aa$ is a unital $C^\ast$-algebra, then the map $\omega \mapsto \Nn_\omega$ from ${\rm PS}(\Aa)$ to the set of maximal closed left ideals is a bijection.
\end{theorem}

Furthermore:

\begin{proposition}[Th.~5.3.3 \cite{MurphyBook}]\label{Prop:LSet}
   For any proper closed left ideal $L$, the set 
    \begin{equation}\label{Eq:Rr}
        \Ll= \{\Nn_\omega  \, | \, \omega \in {\rm PS}(\Aa), \ L \subseteq \Nn_\omega\}
    \end{equation}
   is non-empty and $L = \cap \, \Ll$.
\end{proposition}

    A $C^\ast$-subalgebra $\Bb$ of a $C^\ast$-algebra $\Aa$ is called hereditary if $0\leq a \leq b$ implies $a \in \Bb$ for any $a\in \Aa$ and $b \in \Bb$. We recall the connection between closed left ideals and hereditary $C^\ast$-subalgebras:

\begin{theorem}[\cite{MurphyBook}, Th. 3.2.1]\label{Th:LvsB} Let $\Aa$ be a unital $C^\ast$-algebra. Then
\begin{equation}
    L \mapsto \Bb=L^\ast \cap L, \quad \Bb \mapsto L=\{a \in \Aa \; | \; a^\ast a \in \Bb \}
\end{equation}
is an isomorphism between the poset of proper left closed ideals and the poset $\mathfrak H (\Aa)$ of proper hereditary $C^\ast$-subalgebras of $\Aa$.
\end{theorem}

Let $\hat{\mathfrak H} (\Aa):= {\mathfrak H}(\Aa) \cup \Aa$ be the set of all hereditary $C^\ast$-subalgebras of $\Aa$. When equipped with the partial order given by inclusions, it becomes a sublattice of the lattice of $C^\ast$-subalgebras of $\Aa$, with the infimum given by intersection. This lattice relates to several other interesting lattices appearing in the $C^\ast$-context (see \cite{AkemannAM2015} for a concise list). Given Proposition~\ref{Prop:LSet} and the close relation between left ideals and hereditary subalgebras, we point out the relation most relevant to us:

\begin{proposition}\label{Prop:BOmega}
    If $\Aa$ is a unital $C^\ast$-algebra, then the map $\omega \mapsto \Bb_\omega$ from ${\rm PS}(\Aa)$ to the set of maximal proper hereditary $C^\ast$-subalgebras is a bijection and the map
    \begin{equation}\label{Eq:BtoOmega}
        \Bb \overset{\mathfrak j}{\mapsto} \Omega_\Bb :=\{\omega \in {\rm PS}(\Aa)\, |\, \Bb \subseteq \Bb_\omega :=\Nn_\omega^\ast \cap \Nn_\omega \}
    \end{equation}
    is an injective lattice morphism from $(\hat{\mathfrak H}(\Aa),\subseteq)$ to $(\mathfrak P[{\rm PS}(\Aa)],\supseteq)$, where $\mathfrak P$ denotes the power set. The inverse of this map on its range is
    \begin{equation}
        \Omega \mapsto \Bb_\Omega = \bigcap_{\omega \in \Omega} \Bb_\omega.
    \end{equation}
\end{proposition}

\begin{remark}
    {\rm Note that the above map takes $\Aa\in \hat{\mathfrak H}(\Aa)$ into $\emptyset \in \mathfrak P[{\rm PS}(\Aa)]$ and $\{0\}\in \hat{\mathfrak H}(\Aa)$ into ${\rm PS}(\Aa) \in \mathfrak P[{\rm PS}(\Aa)]$.} $\Diamond$
\end{remark}

\begin{corollary}\label{Cor:StateSets}
    There is a well defined map ${\rm S}(\Aa) \overset{\mathfrak s}{\to} \mathfrak P[{\rm PS}(\Aa)]$ given by the composition
    \begin{equation}
        \omega {\mapsto} \Bb_\omega = \Nn_\omega^\ast \cap \Nn_\omega \mapsto \Omega_{\Bb_\omega}.
    \end{equation}
    We call $\mathfrak s(\omega)$ the hereditary footprint of $\omega$ in ${\rm PS}(\Aa)$.
\end{corollary}

The following characterization of hereditary $C^\ast$-subalgebras is useful: 

\begin{theorem}[\cite{MurphyBook}, Th. 3.2.5]\label{Th:HCharacterization} Let $\Aa$ be a separable $C^\ast$-algebra. If $\Bb$ is a hereditary $C^\ast$-subalgebra of $\Aa$, then there exists a positive element $w_\Bb \in \Bb$ such that 
\begin{equation}
 \Bb = \overline{w_\Bb \Aa w_\Bb} \ {\rm and} \ L(\Bb)=\overline{\Aa w_\Bb}.
 \end{equation}
 \end{theorem}

 \begin{remark}\label{Re:W}
 {\rm We recall the construction of $w_\Bb$. Since $\Aa$ is separable, $\Bb$ is separable too. Hence, the latter accepts a sequential approximate unit $\{1_n^\Bb\}_{n\in \NM^\times}$, and we can choose
 \begin{equation}\label{Eq:WB}
     w_\Bb:=\sum 2^{-n}\, 1_n^\Bb.
     \end{equation}
     Note that we made the choice to have $w_\Bb  \leq 1$. While $w_\Bb$ depends on the chosen approximate unit, $\overline{w_\Bb \Aa w_\Bb}$ does not. If $\Aa$ is finite dimensional, then $w_\Bb$ can be chosen in a {\it unique} way to be a projection from $\Aa$.} $\Diamond$
 \end{remark}

\section{Frustration-freeness vs hereditary $C^\ast$-subalgebras}\label{Sec:FFfromB}

This section establishes the horizontal part of the diagram~\ref{Eq:MasterDiag}. In the process, we introduce and characterize the entry $\mathfrak H^F(\Ss^{\otimes G})$ of this diagram.

\subsection{Intrinsic characterization of frustration-freeness}

Hereditary $C^\ast$-subalgebras of AF algebras are AF themselves, and AF-algebras have real rank zero. Hence, all $\Bb \in \mathfrak H(\Ss^{\otimes G})$ have real rank zero and, as such, they always have approximate units given by nets of projections \cite{StrungBook}. Below, we identify the class of hereditary $C^\ast$-subalgebras for which some of these approximate units are frustration-free systems of projections:

\begin{lemma}\label{Prop:BbUnit}
    Let $\Bb \in \mathfrak H(\Ss^{\otimes G})$ be proper. Then $\Bb \cap \Ss_\Lambda$ is a proper hereditary subalgebra of $\Ss_\Lambda$ for each $\Lambda \in {\rm K}(G)$. As such, there exists a unique projection $p_\Lambda(\Bb) \in \Ss_\Lambda \subset \Ss^{\otimes G}$ such that 
    \begin{equation}
        \Bb \cap \Ss_\Lambda = p_\Lambda(\Bb) \, \Ss_\Lambda \, p_\Lambda(\Bb).
    \end{equation}
    If $\Bb$ displays the property
     \begin{equation}\label{Eq:PropF}
       {\rm (F):} \ \  \Bb = \overline{\cup_{\Lambda \in {\rm K}(G)} (\Bb \cap \Ss_\Lambda)},
\end{equation}
then $\{p_\Lambda(\Bb)\}_{\Lambda \in {\rm K}(G)}$ is an approximate unit for $\Bb$ and a frustration-free proper system of projections.
\end{lemma}

\begin{proof} 
    Property (F) assures us that the intersections $\Bb \cap \Ss_\Lambda$ are not all void and that none of the $p_\Lambda$'s can be the identity. Otherwise, $\Bb$ will contain the unit and will not be proper. Hence, $\{p_\Lambda\}$ is non trivial and proper. 
    We first claim that $\{p_\Lambda\}=\{p_\Lambda(\Bb)\}_{\Lambda \in {\rm K}(G)}$ is an increasing net. Indeed, if $\Lambda_1 \subseteq \Lambda_2$,
    \begin{equation}
        \Bb \cap \Ss_{\Lambda_1} \subseteq \Bb \cap \Ss_{\Lambda_2},
    \end{equation}
    hence $p_{\Lambda_1}= p_{\Lambda_2} s \, p_{\Lambda_2}$ for some $s \in \Ss_{\Lambda_2}$ and, as such, $p_{\Lambda_2} p_{\Lambda_1}=p_{\Lambda_1}$ and the claim follows. Property (F) also assures us that, for each $b\in \Bb$, there exists a net $b_\Lambda \in \Bb \cap \Ss_\Lambda$ such that $\lim b_\Lambda =b$. Then
\begin{equation}
    \|b - p_\Lambda b\| = \|b- b_\Lambda - p_\Lambda(b-b_\Lambda)\| \to 0,
\end{equation}
which shows that $\{p_\Lambda\}$ is indeed an approximate unit.
\end{proof}

\begin{remark}
    {\rm Clearly, if $\Bb$ is $G$-invariant, then $\{p_\Lambda(\Bb)\}_{\Lambda \in {\rm K}(G)}$ is a $G$-system.}\ $\Diamond$
\end{remark}

We denote the subset of proper hereditary $C^\ast$-subalgebras displaying property (F) by $\mathfrak H^F(\Ss^{\otimes G})$ and note that this subset is $G$-invariant. We denote by $\hat{\mathfrak H}^F(\Ss^{\otimes G})$ the subset of all hereditary $C^\ast$-subalgebras displaying property (F). This is mostly for notational purposes, because the latter is just ${\mathfrak H}^F(\Ss^{\otimes G})$ with $\Ss^{\otimes G}$ added to it.

\begin{theorem}\label{Prop:HvsFf}
The following hold:
\begin{enumerate}
    \item The map 
    \begin{equation}\label{Eq:BtoP}
       \mathfrak H^F(\Ss^{\otimes G}) \ni \Bb \overset{\mathfrak p}{\mapsto} \{p_\Lambda(\Bb)\} \in \mathfrak F(\Ss^{\otimes G})
    \end{equation}
    is an injective and $G$-equivariant morphism of posets.

    \item There exists a $G$-equivariant morphism $\mathfrak p^{-1}$ of posets from $\mathfrak F(\Ss^{\otimes G})$ to $\mathfrak H^F(\Ss^{\otimes G})$, which is a left inverse for \eqref{Eq:BtoP}.
    
\end{enumerate}
    
\end{theorem}

\begin{proof}
   (1) The map respects the partial orders because, if $\Bb \subseteq \Bb'$, then $\Bb \cap \Ss_\Lambda \subseteq \Bb' \cap \Ss_\Lambda$ and, as such, $p_\Lambda \in p'_\Lambda \Ss_\Lambda p'_\Lambda$, which implies $p_\Lambda p'_\Lambda = p'_\Lambda p_\Lambda = p_\Lambda$. Furthermore, if two hereditary subalgebras admit identical approximate units, then their corresponding $w$'s coincide and \ref{Th:HCharacterization} assures us that the hereditary $C^\ast$-subalgebras are identical. Now, for each $\Lambda \in {\rm K}(G)$, we have
    \begin{equation}
    \alpha_g(\Bb) \cap \Ss_{g\cdot \Lambda} = p_{g\cdot \Lambda}\big (\alpha_g(\Bb)\big) \, \Ss_{g\cdot \Lambda} \, p_{g\cdot \Lambda}\big (\alpha_g(\Bb)\big ).
    \end{equation}
   In the same time,
   \begin{equation}
     \begin{aligned}
     \alpha_g(\Bb) \cap \Ss_{g\cdot \Lambda} & = \alpha_g(\Bb \cap \Ss_{\Lambda}) \\
     & =\alpha_g\big(p_\Lambda(\Bb)\, \Ss_\Lambda\, p_\Lambda(\Bb)\big )= \alpha_g\big(p_\Lambda(\Bb)\big )\, \Ss_{g \cdot \Lambda} \, \alpha_g\big(p_\Lambda(\Bb)\big ).
     \end{aligned}
    \end{equation}
   Using~\ref{Re:W}, we conclude that
   \begin{equation}
       p_{g\cdot \Lambda}\big (\alpha_g(\Bb)\big)=\alpha_g\big(p_\Lambda(\Bb)\big ) \ \Leftrightarrow \ \{p_\Lambda\big (\alpha_g(\Bb)\big )\} = g \cdot \{p_\Lambda(\Bb)\}
   \end{equation}
   and $G$-equivariance is proved.

   (2) Given $\{p_\Lambda\} \in \mathfrak F(\Ss^{\otimes G})$, we can construct $w=\sum_{m=1}^\infty 2^{-m} p_{\Lambda_m}\in \Ss^{\otimes G}$ for an increasing sequence of finite subsets such that $\cup \Lambda_m = G$. In turn, this supplies the hereditary $C^\ast$-subalgebra $\Bb=\overline{w \Ss^{\otimes G} w}$. We want to show that $\{p_{\Lambda}\}$ is an approximate unit for $\Bb$. Using frustration-freeness,
    \begin{equation}
        p_{\Lambda}w= \sum_{m\leq n} 2^{-m}p_{\Lambda_m}+ \sum_{m > n}2^{-m} p_\Lambda p_{\Lambda_m}= w- \sum_{m > n}2^{-m} (1-p_{\Lambda})p_{\Lambda_m}
    \end{equation}
for any $\Lambda \supseteq \Lambda_n$, hence 
    \begin{equation}
        \|w-p_{\Lambda}w\|\leq 2^{-n}.
    \end{equation}
    For any $b\in \Bb$ and $\epsilon >0$, we can find $a_\epsilon \in \Ss^{\otimes G}$ such that  $\|b-wa_\epsilon w\|\leq \epsilon$. Then
    \begin{equation}
        \| b - p_{\Lambda}b\|\leq \|wa_\epsilon w-p_{\Lambda}wa_\epsilon w\|+2\epsilon \leq 2^{-n}\|a_\epsilon\|\|w\|+2\epsilon.
    \end{equation}
    Hence, there exists $n(\epsilon)$ such that $\| b - p_{\Lambda}b\|\leq 3\epsilon$ for all $\Lambda \supseteq \Lambda_{n(\epsilon)}$. Since $\epsilon$ was arbitrary, $\{p_{\Lambda}\}$ is indeed an approximate unit for $\Bb$. An important outcome of this conclusion is that $\Bb$ so constructed is independent of the choice of $\{\Lambda_m\}$ used in the first place. 

    Let us also verify that $\Bb$ has property (F): For general $b\in \Bb$ and any $\epsilon>0$ we can find some $\Lambda$ large enough such that there exists $a \in \Ss_{\Lambda}$ with $\lVert b- w a w\rVert < \frac{\epsilon}{2}$. By the approximate unit property there is then $n$ large enough such that 
$$\lVert b- p_{\Lambda_n} w a w p_{\Lambda_n} \rVert \leq \lVert b- w a w \rVert + \lVert w a w - p_{\Lambda_n} w a w p_{\Lambda_n} \rVert < \epsilon.$$
Because eventually $\Lambda\subset \Lambda_n$ and then $p_{\Lambda_n} w a w p_{\Lambda_n} \in \Bb\cap \Ss_{\Lambda_n}$, we conclude that property (F) holds.

    Assume now that $\Bb$ is the whole algebra $\Ss^{\otimes G}$, {\it i.e.} that it is not proper. Then $p_{\Lambda}^\bot \in \Bb$ for each $\Lambda\in {\rm K}(G)$ and, since $\{p_\Xi\}$ is an approximate unit for $\Bb$, for each $\epsilon>0$ there exists $\Delta(\epsilon) \supset \Lambda$ such that $\|p_{\Xi} p_{\Lambda}^\bot- p_{\Lambda}^\bot\|\leq \epsilon$ for all $\Xi \supset \Delta(\epsilon)$. By using frustration-freeness, we find
    \begin{equation}
        \|p_{\Xi} p_{\Lambda}^\bot- p_{\Lambda}^\bot\|=\|(1-p_{\Xi}^\bot) p_{\Lambda}^\bot- p_{\Lambda}^\bot\|=\|p_{\Xi}^\bot\|,
    \end{equation}
    and conclude that $\|p_\Xi^\bot \|\leq \epsilon$. Taking $\epsilon <1$, this can only happen if $p_\Xi=1$, but this contradicts our assumptions. Thus, $\Bb$ must be proper and, as such, in $\mathfrak H^F(\Ss^{\otimes G})$.

    Lastly, it is clear that the map respects the partial orders, it is a left inverse for \eqref{Eq:BtoP} and, if $\Bb$ is the image of $\{p_\Lambda\}$, then $\alpha_g(\Bb)$ is the image of $g \cdot \{p_\Lambda\}$, so the map is indeed $G$-equivariant.
\end{proof}

\begin{remark}
{\rm The hereditary $C^\ast$-subalgebra
$\mathfrak p^{-1}(\{p_\Lambda\})$ is precisely the smallest hereditary
$C^\ast$-subalgebra of $\Ss^{\otimes G}$ containing all the projections
$p_\Lambda$. In other words, it is the hereditary $C^\ast$-subalgebra
generated by $\{p_\Lambda\}$.}
\end{remark}
\begin{remark}\label{Re:Ext1}
    {\rm We can extend $\mathfrak p$ to a morphism between the posets $\hat{\mathfrak H}^F(\Ss^{\otimes G})$ and $\hat{\mathfrak F}(\Ss^{\otimes G})$, by declaring $\Ss^{\otimes G} \mapsto \{1_\Lambda\}$. We can also extend $\mathfrak p^{-1}$ to a morphism between the posets $\hat{\mathfrak F}(\Ss^{\otimes G})$ and $\hat{\mathfrak H}^F(\Ss^{\otimes G})$, by declaring that all non-proper systems of projections are mapped into $\Ss^{\otimes G}$. This extension is a left inverse for the extension of $\mathfrak p$.} $\Diamond$
\end{remark}

The morphism $\mathfrak p^{-1}$ clearly cannot be injective. However:

\begin{proposition}
If $\mathfrak F_\Bb$ is the pre-image of the hereditary $C^\ast$-algebra $\Bb$ for the map $\mathfrak p^{-1}$, then $\mathfrak F_\Bb$ has a maximum, which is precisely $\{p_\Lambda(\Bb)\}$.
\end{proposition}

\begin{proof}
Let $\{p_\Lambda \}$ be from $\mathfrak F_\Bb$. From the construction of $\Bb=\mathfrak p^{-1}(\{p_\Lambda\})$, we can see that each $p_\Lambda$ belongs to $\Bb$ and, obviously, also to $\Ss_\Lambda$. Thus
\begin{equation}
    p_\Lambda \in \Bb \cap \Ss_\Lambda = p_\Lambda(\Bb)\, \Ss_\Lambda \, p_\Lambda(\Bb).
\end{equation}
As such, there exists $s \in \Ss_\Lambda$ such that $p_\Lambda=p_\Lambda(\Bb) \, s \, p_\Lambda(\Bb)$ and, as a consequence, $p_\Lambda p_\Lambda(\Bb)=p_\Lambda$ for all $\Lambda \in {\rm K}(G)$. Since $\{p_\Lambda(\Bb)\} \in \mathfrak F_\Bb$, this shows that $\{p_\Lambda(\Bb)\}$ is indeed the maximal element.
\end{proof}

\begin{corollary}
    There exists a $G$-equivariant bijective map between the maximal elements of $\mathfrak H^F(\Ss^{\otimes G})$ and of $\mathfrak F(\Ss^{\otimes G})$.
\end{corollary}

The subset $\hat{\mathfrak H}^F(\Ss^{\otimes G}) \subseteq \hat{\mathfrak H}(\Ss^{\otimes G})$ is not closed under the sup and inf operations in $\hat{\mathfrak H}(\Ss^{\otimes G})$, hence it cannot be called a sublattice of $\hat{\mathfrak H}(\Ss^{\otimes G})$ \cite[p.~20]{GratzerBook}. Nevertheless:

\begin{proposition} $(\hat{\mathfrak H}^F(\Ss^{\otimes G}),\subseteq)$ is a lattice.
\end{proposition}

\begin{proof}
   We are going to exploit Proposition~\ref{Prop:FLattice}, where it was established that $\hat{\mathfrak F}(\Ss^{\otimes G})$ is a lattice. If $\Bb_i \in \hat{\mathfrak H}^F(\Ss^{\otimes G})$ for $i=1,2$, then 
    \begin{equation}
        \{p_\Lambda(\Bb_i)\} \geq \{p_\Lambda(\Bb_1)\} \wedge \{p_\Lambda(\Bb_2)\} \geq \{p_\Lambda(\Bb)\}
    \end{equation}
    for any $\Bb \in \hat{\mathfrak H}^F(\Ss^{\otimes G})$ with $\Bb \subseteq \Bb_i$, $i=1,2$. Then the extension of the morphism $\mathfrak p^{-1}$ of posets applied to $\{p_\Lambda(\Bb_1)\} \wedge \{p_\Lambda(\Bb_2)\}$ supplies a hereditary $C^\ast$-subalgebra $\Bb_1 \wedge_F \Bb_2 \in \hat{\mathfrak H}^F(\Ss^{\otimes G})$, which is larger or equal then any $\Bb \subseteq \Bb_i$, $i=1,2$. Similarly, the same morphism applied to $\{p_\Lambda(\Bb_1)\} \vee \{p_\Lambda(\Bb_2)\}$ supplies a hereditary $C^\ast$-subalgebra $\Bb_1 \vee_F \Bb_2 \in \hat{\mathfrak H}^F(\Aa)$, which is smaller or equal then any $\Bb \supseteq \Bb_i$, $i=1,2$.
\end{proof}

In conclusion, the frustration-freeness was found to derive from the particular filtration $\Ss^{\otimes G} = \overline{\bigcup \Ss_\Lambda}$  by finite-dimensional $C^\ast$-algebras, and it can be entirely captured by property (F). We can formulate the concept beyond the present context as follows:

\begin{definition}[General frustration-freeness]
    Let $\Aa$ be an AF-$C^\ast$-algebra and $\{\Aa_i\}_{i\in I}$ a particular filtration of it by finite-dimensional $C^\ast$-algebras indexed by the countable directed set $I$. A proper hereditary $C^\ast$-subalgebra $\Bb \subset \Aa$ is said to display frustration-freeness relative to $\{\Aa_i\}$ if $\Bb = \overline{\bigcup_{i\in I} (\Bb \cap \Aa_i)}$.
\end{definition}

With this definition in place, we can generate via \ref{Prop:HvsFf} frustration-free systems of projections, {\it i.e.} increasing nets $\{p_i\}_{\i \in I}$ of proper projections such that $p_i \in \Aa_i$,  as well as frustration-free models, all in the general context of AF-algebras. Many of the statements from this paper remain true in this general setting.

\subsection{Density statements}

 We address here the question of how large is the set $\mathfrak H^F(\Ss^{\otimes G})$ inside the set $\mathfrak H(\Ss^{\otimes G})$ of all hereditary $C^\ast$-subalgebras. For this, we recall yet another special property of hereditary $C^\ast$-subalgebras:

\begin{proposition}[\cite{EliottJFA1976}]\label{Prop:Inlcusion}
    Let $\Aa$ be a separable AF algebra and $\Bb \in \mathfrak H(\Aa)$. Then $\Bb$ is an AF algebra and there is an ascending sequence of finite dimensional subalgebras $\Aa_n$ with norm-dense union in $\Aa$ such that $\Bb_n = \Bb \cap \Aa_n$ is an ascending sequence with norm-dense union in $\Bb$.
\end{proposition}
\begin{proof}
 From Proposition~3.1 and Remark~3.2 in \cite{EliottJFA1976}, for any ascending sequence of finite-dimensional sub-algebras $\tilde \Bb_n \subset \Bb$ with norm-dense union, there exists an ascending sequence of finite-dimensional sub-algebras $\Aa_n \subset \Aa$ with norm-dense union and 
 \begin{equation}
 \Bb \cap (\cup_n \Aa_n) =\cup_n \tilde \Bb_n.
 \end{equation}
 Taking $\Bb_n = \Bb \cap \Aa_n$, we have
 \begin{equation}
     \cup_n \Bb_n = \cup_n (\Bb \cap \Aa_n) = \Bb \cap (\cup_n \Aa_n) = \cup_n \tilde \Bb_n
 \end{equation}
 and the statement follows.
\end{proof}

\begin{remark}
    {\rm Proposition~3.1 in \cite{LazarCJM1986} makes the stronger statement than ours, namely, that any filtration $\Bb_n$ of $\Bb$ can be lifted to a filtration $\Aa_n$ of $\Aa$. However, we were not able to produced a complete proof of this stronger statement.} \; $\Diamond$
\end{remark}
Let us also state the following general fact about AF algebras, which we formulate for our specific context:

\begin{proposition}[\cite{DavidsonBook}, Th. III.3.5]\label{Th:Dav}
Suppose $\Aa_n$ is an increasing sequence of finite sub-$C^\ast$-algebras of $\Ss^{\otimes G}$ such that
$\Ss^{\otimes G} = \overline{\cup_n \Aa_n}$.
Then, for any $\epsilon >0$, there exists a unitary element $u \in \Ss^{\otimes G}$ with $\|u-1\|<\epsilon$ so that $\cup_n \Aa_n = u^\ast(\cup_\Lambda \Ss_\Lambda)u$.
\end{proposition}

A direct consequence of the mentioned facts is that any hereditary sub-$C^\ast$-algebra of $\Ss^{\otimes G}$ can be twisted by an inner automorphism to fit into $\mathfrak H^F(\Ss^{\otimes G})$. More precisely:
 
\begin{proposition}\label{Prop:BU} 
   For every $\epsilon >0$ and $\Bb \in \mathfrak H(\Ss^{\otimes G})$, there exists a unitary element $u \in \Ss^{\otimes G}$ such that $\|u-1\| < \epsilon$ and
   \begin{equation}
     \Bb_u:=u \Bb u^\ast= \overline{\cup_{\Lambda \in {\rm K}(G)} (\Bb_u \cap \Ss_\Lambda)}.
    \end{equation}
\end{proposition}

\begin{proof} Let $\Aa_n$ be the sequences from \ref{Prop:Inlcusion} and $\Bb_n = \Aa_n \cap \Bb$ such that $\Bb = \overline{\cup_n \Bb_n}$. For $\epsilon \in \RM_+$, let $u \in \Ss^{\otimes G}$ be the unitary element from \ref{Th:Dav}. Then
\begin{equation}
   u^\ast( \cup_\Lambda (\Bb_u \cap \Ss_\Lambda))u = \Bb \cap u^\ast( \cup_\Lambda \Ss_\Lambda)u = \Bb \cap ( \cup_n \Aa_n) = \cup_n(\Bb \cap \Aa_n ).
\end{equation}
Therefore
\begin{equation}
   \overline{\cup_\Lambda (\Bb_u \cap \Ss_\Lambda)}= u \, \big (\overline{\cup_n \Bb_n} \big )\, u^\ast =u \Bb u^\ast
\end{equation}
and the statement follows.
\end{proof}

\begin{remark}\label{Re:GInv}
    {\rm The situation is very different if one insists on $G$-invariance. Indeed, since $\Ss^{\otimes G}$ displays asymptotic abelianness, there is no element in it that is invariant under the $G$-action. Thus, the above mechanism fails to produce $G$-invariant hereditary $C^\ast$-subalgebra, even if the original one is $G$-invariant.} \ $\Diamond$
\end{remark}

\section{States vs frustration-free systems}\label{Sec:FFvsStates}

This section analyzes the relation between states and frustration-free systems of projections via the bottom part of diagram \eqref{Eq:MasterDiag}. The entry $\mathfrak P^F[{\rm PS}(\Ss^{\otimes G})] \subset \mathfrak P[{\rm PS}(\Ss^{\otimes G})]$ in \eqref{Eq:MasterDiag} can be identified via Proposition~\ref{Prop:BOmega} as the image of $\mathfrak H^F(\Ss^{\otimes G})$ through the lattice morphism $\mathfrak j$ from \eqref{Eq:BtoOmega}.

Reference \cite{AffleckCMP1988} introduced the concept of frustration-free states, which will be the focus of our investigation. We reformulate it as follows:

\begin{definition}[\cite{AffleckCMP1988}]
    Let $\{p_\Lambda\}$ be a frustration-free system of proper projections. One says that $\omega \in {\rm S}(\Ss^{\otimes G})$ is a frustration-free ground state for $\{p_\Lambda\}$ if $\omega(p_\Lambda)=0$ for all $\Lambda \in {\rm K}(G)$. In general, we say that $\omega$ is a frustration-free state if it so for at least one frustration-free system of proper projections.
\end{definition}

\begin{remark}
    {\rm Let $h_\Lambda = \sum_{g \cdot \Delta \subseteq \Lambda} \alpha_g(p_\Delta)$ be the seeds of a frustration-free inner-limit derivation $\delta_{\bm h}$ associated with the frustration-free $G$-system $\{p_\Lambda\}$. If $\omega$ is a frustration-free ground state for $\{p_\Lambda\}$, then each $p_\Lambda$ belongs to $\Bb_\omega$. As a consequence, $\omega(\delta_{\bm h}(a))=0$ for all $a$ in the domain of $\delta_{\bm h}$, hence $\omega$ is invariant to the time evolution generated by $\delta_{\bm h}$, and, furthermore,
    \begin{equation}\label{Eq:GS0}
        \omega\big(a^\ast \delta_{\bm h}(a)\big) = \omega(a^\ast [h_\Xi, a]) = \omega(a^\ast h_\Xi a) \geq 0
    \end{equation} 
    holds for any $a\in \bigcup \Ss_\Lambda$ and for some $\Xi \in {\rm K}(G)$. Since $\bigcup \Ss_\Lambda$ is dense in $Dom(\delta_{\bm h})$, a frustration-free ground state is also a ground state of $\delta_{\bm h}$ in the usual sense.
    } $\Diamond$
\end{remark}

The following statement gives a characterization of the frustration-free $G$-invariant ground states:

\begin{proposition}\label{Prop:FFandGS}
   If the hereditary $C^\ast$-subalgebra $\Bb_\omega$ of a state $\omega \in {\rm S}(\Ss^{\otimes G})$ is $G$-invariant, proper and displays property {\rm (F)}, then $\omega$ is a frustration-free ground state for a frustration-free $G$-system of projections. 
\end{proposition}

\begin{proof}
By Lemma~\ref{Prop:BbUnit}, $\Bb_\omega$ admits a $G$-invariant localized approximate unit $\{p_\Lambda(\Bb_\omega)\}$, which is proper, frustration-free and $\omega(p_\Lambda(\Bb_\omega))=0$.
\end{proof}

A frustration-free system of projections determines a canonical subset of ${\rm PS}(\Ss^{\otimes G})$ via the composition $\mathfrak j \circ \mathfrak p^{-1}$ of maps defined in \ref{Prop:BOmega} and \ref{Prop:HvsFf}. The following statements clarify the relation between these subsets and the frustration-free ground states.

\begin{lemma}\label{Lemm:PLvsOmega}
    Let $\{p_\Lambda\}$ be a frustration-free system of proper projections. Then $\omega$ is a frustration-free ground state for $\{p_\Lambda\}$ if and only if 
    \begin{equation}\mathfrak p^{-1}(\{p_\Lambda\}) \subseteq  \mathfrak \Bb_\omega.
    \end{equation}
\end{lemma}

\begin{proof}
As in \ref{Prop:HvsFf}, define $w=\sum_{m=1}^\infty 2^{-m} p_{\Lambda_m}$ for an increasing sequence $\Lambda_m \in {\rm K}(G)$ such that $\cup \Lambda_m =G$. Then $\mathfrak p^{-1}(\{p_\Lambda\})=\overline{w \Ss^{\otimes G} w}$. If $\omega$ is a frustration-free ground state for $\{p_\Lambda\}$, then $w$ belongs to $\Bb_\omega$ because $\omega(w)=0$, and, as such, $\mathfrak p^{-1}(\{p_\Lambda\}) \subseteq \Bb_\omega$. Reciprocally, if $\mathfrak p^{-1}(\{p_\Lambda\}) \subseteq \Bb_\omega$, then $w \in \Bb_\omega$ and, as such, $\omega(w)=0$ or $\omega(p_{\Lambda_m})=0$ for all $m$. If $\Lambda \in {\rm K}(G)$, then $\Lambda \subseteq \Lambda_m$ for some $m$, hence $\omega(p_\Lambda)=0$ for all $\Lambda \in {\rm K}(G)$ and we can conclude that $\omega$ is a frustration-free ground state for $\{p_\Lambda\}$.
\end{proof}

\begin{theorem}\label{Th:PLvsOmega}
    Let $\{p_\Lambda\}$ be a frustration-free system of proper projections. Then 
    \begin{equation}\label{Eq:PvsB}
        \mathfrak p^{-1}(\{p_\Lambda\})= \bigcap  \mathfrak \Bb_\omega,
    \end{equation}
    where the intersection runs over all frustration-free ground states of $\{p_\Lambda\}$. The statement remains valid if the intersection is restricted to pure
frustration-free ground states.
\end{theorem}

\begin{proof}
    From \ref{Lemm:PLvsOmega}, we have the inclusion $\mathfrak p^{-1}(\{p_\Lambda\}) \subseteq \bigcap  \mathfrak \Bb_\omega$, hence the challenge is to prove the opposite inclusion. Any proper hereditary sub-$C^\ast$-algebra is the intersection of the maximal hereditary sub-$C^\ast$-algebras containing it. According to \ref{Prop:BOmega} and \ref{Lemm:PLvsOmega}, in the case of $\mathfrak p^{-1}(\{p_\Lambda\})$, the latter correspond to the pure states that are frustration-free ground states for $\{p_\Lambda\}$.
\end{proof}

\begin{corollary}\label{Cor:PMinus}
    In the setting of \ref{Th:PLvsOmega}, 
    \begin{equation}
       (\mathfrak j \circ \mathfrak p^{-1})(\{p_\Lambda\}) = \bigcup \mathfrak s(\omega),
    \end{equation}
    where the union runs over all frustration-free ground states for $\{p_\Lambda\}$. The statement remains valid if the union is restricted to pure frustration-free ground states.
\end{corollary}

\begin{proof}
    We recall that, for $\omega \in {\rm S}(\Ss^{\otimes G})$, $\mathfrak s(\omega)=\mathfrak j(\Bb_\omega)$ (see \ref{Cor:StateSets}). Then the statement follows by applying the lattice morphism $\mathfrak j$ from $(\hat{\mathfrak H}(\Aa),\subseteq)$ to $(\mathfrak P[{\rm PS}(\Aa)],\supseteq)$ on \eqref{Eq:PvsB}.
\end{proof}

We now can give a complete and intrinsic characterization of the frustration-free ground states:
\begin{theorem}\label{Th:FFGS}
    A state is a frustration-free ground state if and only if its hereditary footprint is contained in the footprint of a state whose hereditary $C^\ast$-algebra satisfies property (F).
\end{theorem}

\begin{proof}
    If $\omega$ is the ground state for a frustration-free system of projections $\{p_\Lambda\}$, then its support must be contained in the subset $(\mathfrak j \circ \mathfrak p^{-1})(\{p_\Lambda\})\subset {\rm PS}(\Ss^{\otimes G})$. The hereditary $C^\ast$-subalgebra corresponding to this subset satisfies property (F). Reciprocally, if the hereditary $C^\ast$-subalgebra $\Bb_\eta$ of a state $\eta$ satisfies property (F), then \ref{Prop:FFandGS} assures us that $\eta\big(p_\Lambda(\Bb_\eta)\big)=0$. Now, if $\mathfrak s(\omega) \subseteq \mathfrak s (\eta)$, then $\Bb_\eta \subseteq \Bb_\omega$ and, as such, $\omega\big(p_\Lambda(\Bb_\eta)\big)=0$ because $\omega(\Bb_\omega)=\{0\}$.
\end{proof}

Lastly, we introduce a well-known condition that assures that $(\mathfrak j \circ \mathfrak p^{-1})(\{p_\Lambda\})$ consists of a single point:

\begin{proposition}\label{Prop:LTQO}
    Suppose $\{p_\Lambda\}$ is a frustration-free system of proper projections with the property that
    \begin{equation}\label{Eq:LTQO}
     \lim_{\Lambda \supseteq \Delta} \|p_\Lambda^\bot a p_\Lambda^\bot -\omega_\Lambda(a) p_\Lambda^\bot \|_{\BM(\Hh_\Lambda)}=0
    \end{equation}
    for any $\Delta \in {\rm K}(G)$ and local observable $a \in \Ss_\Delta$, where 
    \begin{equation}
        \omega_\Lambda(a)=\frac{1}{{\rm dim}\, p_\Lambda^\bot}{\rm Tr}(a\, p_\Lambda^\bot).
    \end{equation}
Then $\mathfrak p^{-1}(\{p_\Lambda\})$ is maximal.
\end{proposition}

\begin{proof}
   Since the set of frustration-free ground states is a face, Theorem~\ref{Th:PLvsOmega} shows that it is enough to verify that $\{p_\Lambda\}$ accepts a unique frustration-free ground state. Note that \ref{Prop:HvsFf} assures us that $(\mathfrak j \circ \mathfrak p^{-1})(\{p_\Lambda\})$ is not void, that is, $\{p_\Lambda\}$ accepts at least one frustration-free ground state. Let $a\in \Ss_\Delta \subseteq \Ss_\Lambda \simeq \BM(\Hh_\Lambda)$ be a local observable. Then
\begin{equation}\label{eq: FF3}
\left\| p_{\Lambda}^\perp a p_{\Lambda}^\perp - \omega_{\Lambda}(a) p_{\Lambda}^\perp \right\| 
= \sup^{\|\psi\| = 1}_{\psi \in {\rm Ker}(p_{\Lambda})} \left| \langle \psi, a \psi \rangle - \omega_{\Lambda}(a) \right|.
\end{equation}
Consider now two frustration-free ground states $\omega_1$ and $\omega_2$ of $\{p_\Lambda\}$. Then, there exist orthonormal vectors  $\psi_{\Lambda}^{\alpha,j}\in {\rm Ker}(p_{\Lambda})$ and positive coefficients $c_{\alpha,j}$ such that
$\omega_\alpha(a) = \sum_j c_{\alpha,j}\langle \psi_{\Lambda}^{\alpha,j}, a \psi_{\Lambda}^{\alpha,j} \rangle$, with $\sum_j c_{\alpha,j}=1$, $\alpha=1,2$. Using \eqref{eq: FF3}, we have
\begin{equation}
    \begin{aligned}
|\omega_1(a) -\omega_2(a)| &\leq \sum_\alpha |\omega_\alpha(a) - \omega_{\Lambda}(a)| \\
&\leq \sum_{\alpha,j} c_{\alpha,j}\left| \langle \psi_{\Lambda}^{\alpha,j}, a \psi_{\Lambda}^{\alpha,j} \rangle - \omega_{\Lambda} (a) \right|\\
&\leq 2 \left\| p_{\Lambda}^\perp a p_{\Lambda}^\perp - \omega_{\Lambda}(a) p_{\Lambda}^\perp \right\|.
\end{aligned}
\end{equation}
Taking the limit over $\Lambda$ gives $\omega_1(a) =\omega_2(a)$. Since $\Delta$ was arbitrary, this is true for all local observables and, by continuity, $\omega_1 = \omega_2$.
\end{proof}

\begin{remark}
    {\rm Condition \eqref{Eq:LTQO} is a variation of the local topological quantum order (LTQO) in the published literature \cite{NachtergaeleLMP2024}. It is automatic that, if $\{p_\Lambda\}$ is a proper $G$-system that displays LTQO, then its unique frustration-free ground state is pure and $G$-invariant. An optimal version of the LTQO condition is supplied in \ref{Sec:OptLTQO}.} \ $\Diamond$
\end{remark}

\section{Frustration-free open projections}\label{Sec:FFOpenProj}

We start with a few context-free remarks about open projections. Henceforth, let $\Aa$ be again a unital separable $C^\ast$-algebra and $\Aa^{\ast \ast}$ be its double dual. The latter is naturally a $W^\ast$-algebra that can be concretely realized as the weak closure of the universal representation $\pi_u = \bigoplus_{\eta \in {\rm S}(\Aa) } \pi_\eta$ of $\Aa$ \cite{ShermanPICM1950,TakedaPJA1954}. If $\Aa$ is a $G$-$C^\ast$-algebra, then $\Aa^{\ast \ast}$ is a $G$-$W^\ast$-algebra. Unless stated otherwise, all the limits appearing in this section are taken in the weak$^\ast$-topology of $\Aa^{\ast \ast}$, which, on bounded subsets, coincides with the topology induced by the weak operator topology of $\BM(\Hh_u)$.

\begin{definition}[\cite{AkemannJFA1969}]
    A projection $p\in \Aa^{\ast \ast}$ is open if there exists a net $\{a_\alpha\}\in \Aa$ such that $0\leq a_\alpha \uparrow p$. If $p$ is open, one says that $p^\bot : =1-p$ is closed.
\end{definition}

\begin{proposition}[\cite{AkemannJFA1969,AkemannAM2015}]
    We denote the sets of projections and of open projections in $\Aa^{\ast \ast}$ by $\Pp(\Aa^{\ast \ast})$ and $\Pp_o(\Aa^{\ast \ast})$, respectively, and we endow both with the partial order inherited from $\BM(\Hh_u)$. Then $(\Pp(\Aa^{\ast \ast}),\leq )$ and $(\Pp_o(\Aa^{\ast \ast}),\leq )$ are complete lattices.
\end{proposition}

\begin{proof}
    The first statement is always true for $W^\ast$-algebras. Now, let $p,p' \in \Pp_o(\Aa^{\ast \ast})$ and consider the subset 
    \begin{equation}
    J(p,p') =\{q \in \Pp_o(\Aa^{\ast \ast}), \ q \leq p, \ q \leq p' \}.
    \end{equation}
    Then, since the supremum in $\Pp_o(\Aa^{\ast \ast})$ agree with the supremum in $\Pp(\Aa^{\ast \ast})$ \cite[Prop.~II.5]{AkemannJFA1969} (see also \cite[pg.~106]{AkemannAM2015}),
    \begin{equation}
     p \wedge_o p' := \sup\nolimits_{\Pp(\Aa^{\ast \ast})} \, J(p,p')
    \end{equation} 
    exists and is an open projection. Clearly, $p \wedge_o p' \leq p$ and $p \wedge_o p' \leq p'$ and $p \wedge_o p'$ is the maximal open projection with this property. On the other hand, as already stated, $p \vee_o p' = p\vee p'$.
\end{proof}

\begin{remark}
    {\rm Despite the above result, note that, in general, $(\Pp_o(\Aa^{\ast \ast}),\leq )$ is not a sublattice of $(\Pp(\Aa^{\ast \ast}),\leq )$, since $p\wedge p'$ may not be open \cite[Example~II.6]{AkemannJFA1969}. In other words, $\wedge_o \neq \wedge$.} \ $\Diamond$
\end{remark}

\begin{proposition}[\cite{AkemannJFA1969,AkemannJFA1970}]\label{Prop:P1}
    $\Pp_o(\Aa^{\ast \ast})$ and $\mathfrak H(\Aa)$ are isomorphic posets. In one direction, the isomorphism is implemented by 
    \begin{equation}
        \Pp_o(\Aa^{\ast \ast})\ni p \mapsto p\Aa^{\ast \ast} p \cap \Aa = p\Aa p \cap \Aa \in \mathfrak H(\Aa).
    \end{equation}
    and by 
    \begin{equation}
       \mathfrak H(\Aa) \ni \Bb=\overline{w_\Bb \Aa w_\Bb} \mapsto p = \lim_{n\to \infty} w_\Bb^{1/n} \in \Pp_o(\Aa^{\ast \ast})\subset \Aa^{**}
    \end{equation}
    in the opposite direction.
\end{proposition}

\subsection{A density theorem} We now specialize to $\Aa = \Ss^{\otimes G}$.

\begin{proposition}\label{Prop;FtoP}
    There exists a $G$-equivariant morphism of posets
\begin{equation}
   \mathfrak F(\Ss^{\otimes G}) \to \Pp_o((\Ss^{\otimes G})^{\ast \ast}).  
\end{equation}
\end{proposition}

\begin{proof}
A frustration-free system of projections $\{p_\Lambda\} \in \mathfrak F(\Ss^{\otimes G})$ is an increasing net of projections; hence it has a unique strong limit in $\BM(\Hh_u)$, and hence a unique limit in $\Aa^{\ast \ast}$. Then the correspondence $\{p_\Lambda\} \mapsto \lim p_\Lambda$ supplies the morphism mentioned in the statement. This morphism respects the partial orders and the $G$-actions.
\end{proof}

If $\Pp_o^F((\Ss^{\otimes G})^{\ast \ast})$ denotes the image of this morphism, then \ref{Prop;FtoP} supplies the first morphism in the top part of diagram \eqref{Eq:MasterDiag}, while the second morphism in the same part of the diagram derives from \ref{Prop:P1}.

\begin{remark}
    {\rm The subset of open projections induced by the frustration-free systems of projections also coincides with the image of $\mathfrak H^F(\Ss^{\otimes G})$ through the map \ref{Prop:P1}. As a consequence, the set of frustration-free open projections is fully and intrinsically determined by property (F).}  $\Diamond$
\end{remark}

\begin{proposition}
    $\Pp_o^F((\Ss^{\otimes G})^{\ast \ast})$ is dense in $\Pp_o((\Ss^{\otimes G})^{\ast \ast})$ in the norm topology of $\BM(\Hh_u)$.
\end{proposition}

\begin{proof}
    Let $p \in \Pp_o((\Ss^{\otimes G})^{\ast \ast})$ and $\Bb = p \, \Ss^{\otimes G} \, p \cap \Ss^{\otimes G}$ be the associated hereditary sub-$C^\ast$-algebra. Let $\epsilon_n \to 0$ be a sequence of positive numbers and, for each $\epsilon_n$, we pick a unitary element $u_n \in \Ss^{\otimes G}$ as in \eqref{Prop:BU}. Then $u_n \Bb u_n^\ast$ belong to $\mathfrak H^F(\Ss^{\otimes G})$, hence their corresponding open projections $p_n = u_n p u_n^\ast$ belong to $\Pp_o^F((\Ss^{\otimes G})^{\ast \ast})$. The statement now follows because $u_n \to 1$ in $\Aa$ and, since the latter is isometrically embedded in $\BM(\Hh_u)$, $p_n \to p$ in norm topology of $\BM(\Hh_u)$.
\end{proof}

\subsection{The optimal LTQO condition}\label{Sec:OptLTQO} We start with several statements valid for a general unital and separable $C^\ast$-algebra $\Aa$.

\begin{proposition}[\cite{PedersonBook},~Prop.~3.13.6]
    Let $\Bb$ be a maximal hereditary $C^\ast$-subalgebra of $\Aa$, and $p\in \Aa^{\ast \ast}$ be the unique open projection such that $\Bb = p\Aa p \cap \Aa$. Then $p^\perp$ is a minimal closed projection.
\end{proposition}

\begin{proof}
    Suppose the opposite, that there exists $q$ a closed projection strictly below $p^\perp$. Then $q^\perp$ is an open projection strictly above $p$ and, as such, $\Bb$ is strictly contained in $q^\perp \Aa q^\perp \cap \Aa$. This contradicts the assumption that $\Bb$ is maximal.
\end{proof}

\begin{proposition}[\cite{AkemannJFA1969}, Prop.~II.4] The set of minimal closed projections and the set of minimal projections in $\Aa^{\ast \ast}$ coincide.   
\end{proposition}

\begin{corollary}
    The set of maximal hereditary $C^\ast$-subalgebras of $\Aa$, the set of pure states ${\rm PS}(\Aa)$ and the set of minimal projections of $\Aa^{\ast \ast}$ are in bijective relations.
\end{corollary}

 We now specialize the discussion to the context of quantum spin systems.

\begin{lemma}
    A frustration-free system $\{p_\Lambda\}$ of proper projections has a unique frustration-free ground state if and only if $(\lim p_\Lambda)^\perp$ is a minimal projection in $(\Ss^{\otimes G})^{\ast \ast}$.
\end{lemma}

\begin{proof}
    We show first that the set $(\mathfrak j \circ \mathfrak p^{-1})(\{p_\Lambda\}) \subset {\rm PS}(\Ss^{\otimes G})$ must consist of a single point. Indeed, suppose the contrary. Then, according to \ref{Th:FFGS}, every point of $(\mathfrak j \circ \mathfrak p^{-1})(\{p_\Lambda\})$ supplies a distinct pure state which is a frustration-free ground state for $\{p_\Lambda\}$. This contradicts uniqueness. Therefore, $(\mathfrak j \circ \mathfrak p^{-1})(\{p_\Lambda\})$ consists of one point, which automatically implies that the hereditary $C^\ast$-subalgebra corresponding to $\{p_\Lambda\}$ is maximal, hence the corresponding open projection $p=\lim p_\Lambda$ has the stated property.

    Conversely, let $
q:=\left(\lim_\Lambda p_\Lambda\right)^\perp$ 
be minimal. If $\omega$ is frustration-free, then
$\omega(p_\Lambda)=0$ for every $\Lambda$, and hence its support
projection $p_\omega$ satisfies 
$$
0\neq p_\omega\leq\bigwedge_\Lambda p_\Lambda^\perp=q.$$
By minimality, $p_\omega=q$. Since a minimal projection supports a
unique state, every frustration-free state coincides with the state
supported by $q$. 
\end{proof}

\begin{theorem}[Optimal LTQO]
    A frustration-free system $\{p_\Lambda\}$ of proper projections has a unique frustration-free ground state if and only if there exists a state $\omega$ on $\Ss^{\otimes G}$ such that
    \begin{equation}\label{Eq:WeakLTQO}
       {\rm SOT}_{\BM(\Hh_u)}- \lim \big (p_\Lambda^\perp a p_\Lambda^\perp - \omega(a) p_\Lambda^\perp \big )=0,
    \end{equation}
   for all $a \in \Ss^{\otimes G}$. Furthermore, if \eqref{Eq:WeakLTQO} holds, then $\omega$ is the unique frustration-free ground state, which is necessarily pure.
\end{theorem}

\begin{proof}
    Suppose $\{p_\Lambda\}$ has a unique frustration-free ground state and let $p=\lim p_\Lambda$. Then $p^\perp=\lim p_\Lambda^\perp$ and $p^\perp$ is a minimal projection in $(\Ss^{\otimes G})^{\ast \ast}$. The latter is equivalent to the fact that $p^\perp (\Ss^{\otimes G})^{\ast \ast} p^\perp = \CM \cdot p^\perp$ and, as such, for any $a \in \Ss^{\otimes G}$ we must have $p^\perp a p^\perp= \omega(a) p^\perp$. We can check that the association $\Ss^{\otimes G} \ni a \mapsto \omega(a) \in \CM$ is linear, positive and unital, hence a state over $\Ss^{\otimes G}$. Then the direct implication follows, because $p^\perp a p^\perp= \omega(a) p^\perp$ implies \eqref{Eq:WeakLTQO}.

    Assume now that \eqref{Eq:WeakLTQO} holds for all $a\in \Ss^{\otimes G}$. Then $p^\perp a p^\perp = \omega(a) p^\perp$ and this relation holds over $(\Ss^{\otimes G})^{\ast \ast}$ if $\omega$ is canonically extended to a normal state. As such, $p^\perp$ is minimal, which automatically implies that $\{p_\Lambda\}$ has a unique frustration-free ground state. The latter must coincide with $\omega$.
\end{proof}

\begin{remark}
    {\rm We emphasized the strong limit in \eqref{Eq:WeakLTQO} to contrast with the operator norm limit in \eqref{Eq:LTQO}.} \ $\Diamond$
\end{remark}

\subsection{Images in the Cuntz semigroup} The modern formulation of the Cuntz semigroup \cite{CowardJRAM2008}, originally introduced in \cite{CuntzMA1978}, was shown to also classify the open projections relative to a natural comparison relation \cite{OrtegaJFA2011}. If $\Ss = M_2(\CM)$, the Cuntz semigroup of the quasi-local algebra is \cite[Example 4.12]{GardellaEMS2025}
\begin{equation}
   {\rm  Cu}(\Ss^{\otimes G})= \NM[\tfrac{1}{2}] \sqcup (0,\infty].
\end{equation}
We want to calculate the image of $\Pp_o^F((\Ss^{\otimes G})^{\ast \ast})$ under the $\rm Cu$ functor. The following gives a partial answer:

\begin{proposition}
    Let $p = \lim p_\Lambda$ where $\{p_\Lambda\}$ is the frustration free system from \ref{Ex:VBS}. Then $[p]_{\rm Cu} = [1]_{\rm Cu}$.
\end{proposition}

\begin{proof}
    The space $T(\Ss^{\otimes G})$ of semifinite traces reduces to the unique normalized trace $\tau$ on $\Ss^{\otimes G}$. Let $\tilde \tau$ be its unique extension to a normal tracial state on $(\Ss^{\otimes G})^{\ast \ast}$. Then, according to \cite{OrtegaJFA2011}, $p \sim_{\rm Cu} p'$ if and only if $\tilde \tau(p)=\tilde \tau(p')$. We have
    \begin{equation}
        \tilde \tau (p)= \lim \tau(p_\Lambda)= \lim \frac{{\rm Tr}(p_\Lambda)}{{\rm dim} \, \Ss_\Lambda}= 1- \lim \frac{{\rm Tr}(p_\Lambda^\bot)}{{\rm dim} \, \Ss_\Lambda}.
    \end{equation}
    The statement follows because ${\rm Tr}(p_\Lambda^\bot) \leq {\rm dim} \, \Ss_{\Lambda_B}$.
\end{proof}

\begin{remark}
    {\rm According to the above, the Cuntz semigroup cannot distinguish the frustration-free open projections. However, a $G$-equivariant version of the Cuntz group could be an effective tool for classifying $G$-systems of projections, {\it e.g.} based on the conclusions from \ref{Re:GInv}. Equivariant versions of the Cuntz semigroup already exist for compact groups \cite{GardellaPLMS2017}.} \ $\Diamond$
\end{remark}

\vspace{0.2cm}

\noindent {\bf Conflict of Interests.} The authors declare no conflicts of interests.

\vspace{0.2cm}

\noindent {\bf Data Availability.} Our work did not use or generated any data.

\end{document}